\documentclass[a4paper, 11pt]{article}
\usepackage[utf8]{inputenc} %
\usepackage[T1]{fontenc}    %
\usepackage{hyperref}       %
\usepackage{url}            %
\usepackage{booktabs}       %
\usepackage{amsfonts}       %
\usepackage{nicefrac}       %
\usepackage{microtype}      %
\usepackage{xcolor}         %
\usepackage{wrapfig}

\usepackage{color}
\usepackage{tikz}
\usepackage{pgfplots}
\usepackage{caption}
\usepackage{subcaption}
\usepackage{environ}
\usepackage{enumitem}
\usepackage{bm}
\usepackage{amsmath}
\usepackage{amssymb} %
\usepackage{amsthm}
\usepackage{cleveref}
\usepackage[numbers]{natbib} %
\usepackage[ruled,linesnumbered]{algorithm2e} %
\RestyleAlgo{ruled} %
\SetAlgoCaptionLayout{centerline} %
\SetAlgoInsideSkip{smallskip} %

\SetAlFnt{\small}
\SetAlCapFnt{\small}
\SetAlCapNameFnt{\small}
\SetAlCapHSkip{0pt}
\crefname{algocf}{alg.}{algs.}
\Crefname{algocf}{Algorithm}{Algorithms}
\usetikzlibrary{decorations.pathreplacing}
\usetikzlibrary{intersections}
\usepgfplotslibrary{fillbetween}
\allowdisplaybreaks

\usepackage[margin=0.75in]{geometry}

\newcommand{\R}{\mathbb{R}}
\newcommand{\N}{\mathcal{N}}
\newcommand{\A}{\mathcal{A}}

\DeclareMathOperator*{\E}{\mathbb{E}}
\renewcommand{\S}{\mathcal{S}}

\newcommand{\D}{\mathcal{D}}

\newcommand{\OPT}{\mathsf{OPT}}

\newcommand{\eps}{\varepsilon}

\newcommand{\abs}[1]{\left| #1 \right|}

\newtheorem{theorem}{Theorem}[section]
\newtheorem{lemma}[theorem]{Lemma}

\newtheorem{assumption}[theorem]{Assumption}

\newtheorem{proposition}[theorem]{Proposition}

\newtheorem{remark}[theorem]{Remark}

\DeclareMathOperator*{\argmax}{arg\,max}
\DeclareMathOperator*{\argmin}{arg\,min}

\title{Procurement Auctions via Approximately Optimal Submodular Optimization}

\author{%
  \begin{tabular}{@{}c@{\hskip 10pt}c@{\hskip 10pt}c@{}}
    Yuan Deng & Amin Karbasi & Vahab Mirrokni \\
    Google Research & Yale University & Google Research \\
    \texttt{dengyuan@google.com} & \texttt{amin.karbasi@yale.edu} & \texttt{mirrokni@google.com} \\
    \\ %
    Renato Paes Leme & Grigoris Velegkas\thanks{Part of the work was done while the author was a student researcher at Google Research.} & Song Zuo \\
    Google Research & Yale University & Google Research \\
    \texttt{renatoppl@google.com} & \texttt{grigoris.velegkas@yale.edu} & \texttt{szuo@google.com} \\
  \end{tabular}
}

\date{}

\begin{document}

\maketitle

\begin{abstract}
We study the problem of procurement auctions, in which an auctioneer seeks to acquire services from a group of strategic sellers with private costs. The quality of the services is measured through some submodular function that is known to the auctioneer. Our goal is to design computationally efficient procurement auctions that (approximately) maximize the difference between the quality of the acquired services and the total cost of the sellers, in a way that is incentive compatible (IC) and individual rational (IR) for the sellers, and generates non-negative surplus (NAS) for the auctioneer.
{Our contribution is twofold: \textbf{i)} we provide an improved analysis of existing algorithms for non-positive submodular function maximization and \textbf{ii)}  we design computationally efficient frameworks that transform submodular function optimization algorithms to mechanisms that are IC and IR for the sellers, NAS for the auctioneer, and approximation-preserving.}
Our frameworks are general and work both in the offline setting where the auctioneer can observe the bids and the services of all the sellers simultaneously, and in the online setting where the sellers arrive in an adversarial order and the auctioneer has to make an irrevocable decision whether to purchase their service or not. 

We further investigate whether it is possible to convert state-of-the-art submodular optimization algorithms into  descending auctions. We focus on the adversarial setting, meaning that the schedule of the descending prices is determined by an adversary. We show that a submodular optimization algorithm satisfying bi-criteria $(1/2, 1)$-approximation in welfare can be effectively converted to a descending auction in this setting.
We further establish a connection between descending auctions and online submodular optimization.
Finally, we demonstrate the practical applications of our frameworks by instantiating them with different state-of-the-art submodular optimization algorithms and comparing their welfare performance through empirical experiments on publicly available datasets that consist of thousands of sellers.
\end{abstract}

\newpage

\section{Introduction}

In this paper, we consider procurement auctions~\citep{dimitri2006handbook} in which strategic service sellers with private costs submit bids to an auctioneer, who then decides the set of winners based on an objective function and purchases their services. Procurement auctions have been adopted in multitude of application domains, including industrial procurement~\citep{bichler2006industrial}, data acquisition~\citep{roth2012conducting,sim2022data,fallah2023optimal}, and crowdsourcing markets~\citep{singer2013pricing}. Additionally, Amazon Business provides government procurement solutions for products ranging from office supplies to first responder equipment~\citep{amazonbusiness2024}, while the U.S. Government Publishing Office conducts nationwide procurement for items used in publishing~\citep{gpo2024}. Each of these markets involve thousands of vendors for ensuring supplier diversity and competitive prices.

Since procurement auctions were introduced to the algorithmic game theory community in the seminal paper of \citet{nisan1999algorithmic}, many aspects of these auctions have been studied, including frugality in procurement auctions for minimizing purchasing costs~\citep{archer2007frugal,karlin2005beyond,talwar2003price}, budget-feasible procurement auctions in which the purchasing cost is constrained by a budget~\citep{singer2010budget,kempe2010frugal,chen2011approximability,bei2012budget,balkanski2022deterministic}, and profit maximization for optimizing the auctioneer's surplus~\citep{cary2008auctions}. In this paper, we consider the classic setting of procurement auctions with one auctioneer and a group of sellers, $\N$. Each seller has a {\em private} cost $c_i$ for providing the service, and for each $S \subseteq \N$, the auctioneer has a value function $f(S)$ for purchasing services from sellers in $S$. The social welfare obtained from the procurement auction when the auctioneer purchases from sellers in $S$ is given by the difference between the value obtained by the auctioneer and the total cost of sellers in $S$, i.e., $f(S) - \sum_{i \in S} c_i$.

We note that our objective differs from the utilitarian objective, which maximizes the sum of all agents' valuations, i.e., $f(S) + \sum_{i \not\in S} c_i$, assuming seller $i$ has value $c_i$ if the service is not sold. Our objective can be considered a measure similar to the {\em gains-from-trade} in the bilateral trade literature, as it measures the {\em additional} total value generated by running the procurement auction. This is further motivated by the fact that, in many domains where procurement auctions can be applied, sellers incur costs only \emph{after} they have been selected to provide a service. For instance, when a vendor provides goods or services to a company, it is most often the case that they incur costs only after signing the contract and needing to purchase the required materials. Although maximizing gains-from-trade is equivalent to maximizing the utilitarian objective, providing an approximation to the gains-from-trade is considerably harder than providing an approximation to the utilitarian objective. In particular, a non-zero approximation to the utilitarian objective may already be achieved without any trade, which would only result in a zero approximation to gains-from-trade.

Another major obstacle to studying our objective is that maximizing for $f(S) - \sum_{i \in S} c_i$ is a computationally challenging optimization problem, even when $f$ is a \emph{monotone submodular} function.  Submodularity captures a broad class of functions with diminishing returns, including gross-substitute functions, budget additive functions, and coverage functions. When $f$ is a monotone submodular function, the aforementioned optimization task is known as {\em regularized submodular maximization}, expressed as maximizing the difference between a monotone submodular function $f(S)$ and a modular function $\sum_{i \in S} c_i$ as the regularization term. For the special case where $c_i = 0$ for all sellers, it is well-known that obtaining an approximation ratio better than $(1-1/e)$ requires exponentially many queries to $f$~\citep{nemhauser1978analysis}. Before a breakthrough result from \citet{sviridenko2017optimal}, many heuristics had been proposed in the literature~\citep{feige2008combinatorial,feige2013pass,kleinberg2004segmentation}, but none of them provided provable guarantees that are universal and unconditional. For example, \citet{feige2013pass} design algorithms with parameterized approximation guarantees, where the parameterization may depend on classes of instances and the properties of their optimal solutions. \citet{sviridenko2017optimal} show a tight {\em bi-criteria} $(\alpha, \beta)$-approximation guarantee of the form $f(S) - \sum_{i \in S} c_i \geq \alpha \cdot f(\OPT) - \beta \cdot \sum_{i \in \OPT} c_i$ with $\alpha = 1 - 1/e$ and $\beta = 1$, where $\OPT = \argmax_S f(S) - \sum_{i \in S} c_i,$
which rules out the possibility of constant approximation for the welfare objective. \citet{harshaw2019submodular} developed a simplified and practical algorithm, called {\em the distorted greedy algorithm}, that achieves the same optimal approximation guarantee. 

Inspired by recent progress in regularized submodular maximization, we revisit the mechanism design problem of procurement auctions where our objective is to maximize the welfare, i.e., the difference between the value obtained by the auctioneer and the total cost of sellers.
Our goal is to design a mechanism that
\begin{itemize}
    \item is incentive compatible (IC) and individual rational (IR) for each seller;
    \item guarantees non-negative auctioneer surplus (NAS);
    \item achieves state-of-the-art bi-criteria welfare guarantees;
    \item is computationally efficient.
\end{itemize}

The interpretation of the desired game-theoretical properties, i.e., IC, IR, and NAS, 
is that IR encourages sellers to participate in the auction while IC prevents strategic behavior and simplifies the sellers' decision-making process. The surplus of the auctioneer is given by $f(S) - \sum_{i \in S} p_i$ for purchasing services from sellers in $S$, where $p_i$ is the payment to seller $i \in S$. NAS is reminiscent of the weakly budget-balanced property under the interpretation that $f(S)$ captures the potential revenue generated for the auctioneer through the services purchased from sellers in $S$. The NAS property is critical for many applications; if the auctioneer is at risk
of having negative revenue, they
might be incentivized not to run the auction at all.
Thus, it is 
crucial to ensure that our transformations from algorithms
to mechanisms satisfy NAS, which requires
subtle technical work.

\subsection{Our Results}

In this paper, we make theoretical
contributions to the literature of regularized submodular
optimizations
as well as theoretical
and empirical contributions to the literature of procurement auctions.

On the submodular optimization
side, \citet{harshaw2019submodular} show that the distorted greedy algorithm satisfies an $(1 - 1/e, 1)$ bi-criteria approximation guarantee. In Section~\ref{sec:submodular-algos}, we demonstrate that
 the distorted greedy algorithm also satisfies
$\big(1-e^{-\beta}, \beta + o(1)\big)$ bi-criteria 
approximation simultaneously for all $\beta \in [0,1]$, which is also almost tight \citep{feldman2021guess}.
Using the framework developed in this work, we immediately obtain mechanisms that
satisfy the same approximation guarantees for all $\beta \in [0,1]$. We also extend the results to the setting with noisy function evaluations.

Moving on to the mechanism design side in Section~\ref{sec:sealed-bid-framework}, from a theoretical perspective, we first show that VCG mechanisms satisfy IC, IR, and NAS, and they are always welfare-optimal, but it is computationally prohibitive to implement them for practical applications. We then develop a mechanism design framework that can convert state-of-the-art submodular optimization algorithms to sealed-bid mechanisms that satisfy IC, IR and NAS, preserve the bi-criteria welfare guarantees, and can be computed efficiently. Leveraging online submodular optimization algorithms, we extend our framework to the online setting where the sellers arrive in a potentially adversarial order and the auctioneer has to make an irrevocable decision whether to purchase their services or not.

In addition to sealed-bid mechanisms, in Section~\ref{sec:descending-auctions} we ask whether it is possible to convert submodular optimization algorithms to  descending auctions. These auctions were initially designed under the assumption that $f$ is a gross-substitute function~\citep{kelso1982job}, which is a subclass of submodular functions.
We focus on the adversarial setting where the schedule of descending prices is determined by an adversary. %
 We show that if the demand oracle is based on the cost-scaled greedy algorithm~\citep{nikolakaki2021efficient}, 
the descending auction always achieves bi-criteria $(\frac 1 2, 1)$-approximation in welfare, even in the adversarial setting. On the other hand, we show that if the oracle solves the demand problem exactly, the approximation guarantees
could be arbitrarily bad.
We further establish a connection between descending auctions and online submodular optimization algorithms, showing that any online submodular optimization algorithm can be converted to a descending auction in an approximation-preserving way. Thus, an impossibility result, showing there is no descending auction that can achieve bi-criteria $(\alpha,1)$-approximation in welfare with $\alpha > \frac 1 2$, directly implies an impossibility result on online submodular optimization with bi-criteria $(\alpha,1)$-approximation guarantees for the same $\alpha$, which is a long-standing open question in online submodular optimization.

In \Cref{sec:experiment}, we complement our theoretical results with empirical studies evaluating the welfare performance and running time trade-offs of different mechanisms on a publicly available coverage problem. Due to space constraints, the discussions of further related work are deferred to \Cref{apx:related-work}.

\section{Preliminaries}\label{sec:prelim}

We consider a setting of procurement auctions with one auctioneer and a set $\N$ of $n$ sellers with items to sell. The auctioneer has a valuation function $f: 2^{\N} \to \R_{\geq 0}$ that specifies the value that the auctioneer assigns to the items of every set $S$ of sellers, where $S \subseteq \N$. Each seller $i \in \N$ has a private cost $c_i \geq 0$ indicating the minimum acceptable payment for selling to the auctioneer.
We focus on functions $f$ that are monotone and submodular with $f(\emptyset) = 0$. A function $f$ is monotone if $f(S) \leq f(T)$ for all $S \subseteq T \subseteq \N,$ and submodular if it satisfies the property of diminishing returns: $f(i \mid S) >= f(i \mid T)$ for all $S \subseteq T \subset \N$ and $i \not\in T$, where $f(i \mid S) = f(S \cup \{i\}) - f(S)$ computes the marginal contribution of seller $i$ to $f$, conditioned on $S$. Throughout the paper, we use bold symbols $\bm x$ to represent a vector with $n$ elements $(x_1, \cdots, x_n)$ and use $\bm x_{(a,b)}$ to represent $(x_a, \cdots, x_b)$. As usual, we use $-i$ to indicate all the sellers other than seller $i$. 

Let $b_i$ be the reported bid from seller $i$. A mechanism $M = (a, p)$ consists of an allocation rule $a: \R_{\geq 0}^n \to 2^{\N}$ that maps sellers' reported bids $\bm b$ to a subset of sellers to procure the items, and a payment rule $p: \R_{\geq 0}^n \to \R_{\geq 0}^n$ that maps sellers' reported bids $\bm b$ to a vector of payments to each seller. 
We assume sellers have quasi-linear utilities such that given a bid profile $\bm b$ and a mechanism $M$, seller $i$'s utility is given by $u_i^M(\bm b) = p(\bm b) - c_i \cdot \mathrm{1}[i \in a(\bm b)]$. Our goal is to design a  mechanism $M$ that is incentive compatible, individual rational, and induces non-negative auctioneer surplus.
\begin{itemize}
    \item A mechanism is incentive compatible (IC) if it is always an optimal strategy for a seller to report their private cost truthfully, i.e., for any seller $i$ and any $\bm b$, $u_i(c_i, \bm b_{-i}) \geq u_i(\bm b)$.
    \item A mechanism is individual rational (IR) if a seller's utility is always non-negative if they report truthfully, i.e., for any seller $i$ and any $\bm b_{-i}$, $u_i(c_i, \bm b_{-i}) \geq 0$.
    \item A mechanism satisfies the non-negative auctioneer surplus (NAS) condition if the acquired value of the auctioneer is at least the total payment to the sellers, i.e., $f\big( a(\bm b) \big) \geq \sum_{i \in 1}^n p_i(\bm b)$, for any bid profile $\bm b$.
\end{itemize}
We will refer to mechanisms that satisfy the IC, IR, and NAS conditions
as \emph{feasible} mechanisms.
We measure the performance of a mechanism by its welfare $f\big( a(\bm c) \big) - \sum_{i \in a(\bm c)} c_i$ and let $\OPT = \argmax_S f(S) - \sum_{i \in S} c_i,$ when the definition of $f$ is clear from context. We may also write $c(S) = \sum_{i \in S} c_i$. We say that a mechanism satisfies bi-criteria $(\alpha, \beta)$-approximation to the welfare if 
\[
    f\big( a(\bm c) \big) - \sum_{i \in a(\bm c)} c_i \geq \max\left\{0, \alpha \cdot f(\OPT) - \beta \cdot \sum_{i \in \OPT} c_i\right\}\,.
\]

\section{Submodular Optimization Algorithms}
\label{sec:submodular-algos}

We first present several submodular
optimization algorithms that will be useful
for the derivation of our mechanisms, and provide an improved analysis for the
deterministic and stochastic versions of the distorted greedy algorithm \citep{harshaw2019submodular}. 
We will demonstrate how to convert
all the algorithms from this section
to NAS, IC, and IR mechanisms that maintain
the approximation guarantees of the underlying
algorithms in~\Cref{sec:sealed-bid-framework}.

Recall that in the regularized
submodular maximization problem under a cardinality constraint,
there is a monotone submodular function $f: 2^\N 
\rightarrow \R_{\geq 0}$, a cost $c_i \in \R_{\geq 0}$
for each element $i \in \N$, and a cardinality constraint $k \leq n$.
All the algorithms we touch upon in this
section, and subsequently in our mechanism
design framework, share the same paradigm (\Cref{alg:meta-algo}):
the algorithm maintains a candidate solution set initialized as 
$S = \emptyset$ and in each round $k$,
it assigns a \emph{score} to each element $i \in \mathcal{N} \setminus S$ based on a scoring function $G$, which depends on the cost vector $\bm c$ and possibly the round number $k$ as well as a random seed $r$ (for randomized algorithms). The algorithm then adds the element
with the highest non-negative score to $S$.

We next describe each algorithm in detail, focusing on the corresponding scoring function $G$. The bi-criteria
guarantees of all the algorithms
are deferred to \Cref{tab:mechanism-guarantees-different-submodular-algorithm} in \Cref{apx:regularized-submodular-optimization}.
To simplify the notation, when we discuss deterministic algorithms, we omit referring to the random seed that $G$ could take as input, and we omit referring to the round number $k$ as input of $G$ when $G$ does not use the round number information. 

\noindent\textbf{Greedy-margin~\citep{kleinberg2004segmentation}}. We start with the simplest algorithm, called the greedy-margin algorithm, which is
perhaps the most natural approach one could use. 
This algorithm simply chooses the seller with the largest difference between their marginal contribution and their cost, i.e., the scoring function is given by
\[
    G(i, S, \bm c) = f(i \mid S) - c_i.
\]

\noindent\textbf{Greedy-rate~\citep{feige2013pass}}. The greedy-rate algorithm chooses the seller that maximizes the ratio of the difference between their marginal contribution and their cost, over their marginal contribution, i.e., the scoring function is given by
\[
    G(i, S, \bm c) = \frac{f(i \mid S) - c_i}{f(i \mid S)}.
\]

\noindent\textbf{Distorted Greedy~\citep{harshaw2019submodular}}. The distorted greedy algorithm shares a similar flavor to the classical algorithm of 
\citet{nemhauser1978analysis}, but with a slightly \emph{distorted} objective with a multiplier $\left(1-\frac{1}{n}\right)^{n - k}$ on the marginal contribution in round $k$:
\[
    G(i,S,\bm c, k) = \left(1-\frac{1}{n}\right)^{n - k}\cdot f(i \mid S) - c_i \,.
\]
It is worth highlighting that this scoring function does not have the diminishing-return structure and in particular, it does not stop early even if scores are negative for all remaining candidates.

\noindent\textbf{Stochastic Distorted Greedy~\citep{harshaw2019submodular}}. 
In order to speed up the execution of the distorted greedy algorithm, \citet{harshaw2019submodular} proposed a randomized implementation
of it that works as follows. It runs for $n$ iterations
and in every iteration $k$ it draws a seller uniformly
at random from $\N$. Assume the random seed $r$ encodes the selected seller in iteration $k$ via $r(k)$. Then, we can define the scoring function as
\[
    G(i,S,\bm c,k,r) = \mathrm{1}[i = r(k)] \cdot \left(\left(1-\frac{1}{n}\right)^{n - k}\cdot f(i \mid S) - c_i\right) \,.
\]
\citet{harshaw2019submodular} showed that, in expectation over the random draws of the sellers, 
the approximation guarantee of this algorithm does not degrade
compared to its deterministic counterpart.

\noindent\textbf{Return-on-Investment (ROI) Greedy~\citep{jin2021unconstrained}.} The ROI
greedy algorithm
chooses the seller that has the largest marginal
contribution per unit of their cost among sellers whose cost not exceeding their marginal contribution, i.e., the scoring
function is given by
\[
    G(i, S, \bm c) = \frac{f(i \mid S)  - c_i}{c_i} \,.
\]
Observe that, ROI greedy is effectively the same as greedy-rate as both algorithms are effectively ranking the sellers in descending order of $\frac{f(i \mid S)}{c_i}$. \citet{feige2013pass} provide a parameterized approximation guarantee for this algorithm while \citet{jin2021unconstrained} demonstrate a unconditional approximation guarantee %
 without paying a linear term on  $f(\mathrm{OPT})$, which is desirable when this quantity is large.

\noindent\textbf{Cost-scaled Greedy~\citep{nikolakaki2021efficient}}. The cost-scaled greedy algorithm chooses the seller with the largest difference between their marginal contribution and {\em twice} their cost, i.e., the scoring function is given by
\[
    G(i, S, \bm c) = f(i \mid S) - 2 \cdot c_i.
\]
In fact, the cost-scaled greedy algorithm can also be applied to the online and adversarial setting in which the sellers arrive in an online manner (such that any decision is irrevocable) and the sequence of their arrival is determined by an adversary. In the online and adversarial setting, the algorithm maintains a tentative solution $S$ and adds a newly arrived seller to the solution $S$ if and only if $f(i \mid S) - 2 \cdot c_i > 0$.

\subsection{Improved Analysis for Distorted Greedy}
We now explain the improved
analysis we propose for the distorted greedy
algorithm.
Recall that the distorted greedy score of every element $i \in \N$ in every round $1 \leq j \leq k $ of the
execution of the algorithm is $
    \left(1-\frac{1}{n}\right)^{k - j}\cdot f(i \mid S_{j-1}) - c_i \,,
$
and the element that maximizes it is added to the current solution, provided
that its distorted score is non-negative.
\citet{harshaw2019submodular} showed that both versions
of the algorithm enjoy (roughly) a $(1-1/e,1)$-bi-criteria approximation
guarantee, which is tight.
We show that, in fact, these algorithms satisfy an even stronger guarantee: the distorted greedy algorithm enjoys 
$(1-e^{-\beta}, \beta + o(1))$-bi-criteria guarantee for all $\beta \in [0,1]$, where the $o(1)$ term
is sub-constant in cardinality $k$. A similar result holds, in expectation, for the stochastic
version of the algorithm.
The guarantee of the algorithm
holds \emph{simultaneously} for all $\beta \in [0,1],$ so it does not 
require parameterization by $\beta$.

\begin{theorem}\label{thm:bi-criteria-distorted}
    Let $\N$ be a universe of $n$ elements, $f: 2^\N \rightarrow \R_{\geq 0}$
    be a monotone submodular function, and $c: \N \rightarrow \R_{\geq 0}$
    be a cost function. Let $\OPT$ be the optimal solution of the objective
    $\max_{S \subseteq \N, \abs{S} \leq k} \{ f(S) - \sum_{i \in S} c_i\}.$
    Then, the output of the distorted greedy algorithm satisfies $f(R) - \sum_{j \in R}c_j \geq (1-e^{-\beta})f(\OPT) - (\beta + 1/k) \sum_{j \in \OPT} c_j,$ simultaneously
    for all $\beta \in [0,1].$
\end{theorem}

The proof, as well as the formal statement for the
the algorithm, are postponed to \Cref{apx:regularized-submodular-optimization}. 
Our main technical insight is to perform a parameterized \emph{analysis} of
the potential 
function argument of \citet{harshaw2019submodular} based on the target value
of $\beta$ that we wish to prove the guarantee for. In other words, given
some $\beta \in [0,1],$ we lower bound the potential function by a $\beta$-dependent quantity.
This allows us to obtain the stated guarantees for all $\beta$ simultaneously. 
A result of \citet{feldman2021guess} (cf. \Cref{thm:bi-criteria-pareto-lower-bound}) shows that our analysis 
achieves the Pareto frontier of the bi-criteria guarantees for this problem, up to the $o(1)$ term.
Details are deferred to \Cref{apx:regularized-submodular-optimization}. 

In \Cref{apx:regularized-submodular-optimization} we
present an adaptation of the distorted greedy
algorithm that works even when we only have access
to an approximate version $F: 2^\N \rightarrow \R_{\geq 0}$ of valuation function $f$ such that $(1-\eps) f(S) \leq F(S) \leq (1+\eps) f(S), \forall S \subseteq \N$~\citep{horel2016maximization}. \citet{gong2023algorithms} propose
a slight adaptation of the distorted
greedy algorithm that performs well when
$\eps = O(1/k).$ However, 
when we convert their algorithm to a mechanism,
it is not immediate
how to prove the NAS property,
since $F$ might not be submodular, which 
was a crucial property of the function in our
later proof of NAS. Thus, we propose
a modification of their algorithm to overcome this issue (see \Cref{alg:noisy-distorted-greedy}).
Our main insight is to have the greedy scores of
the elements in round $t$ of the execution
depend not only on the current tentative solution $S_t$,
but on the whole trajectory $S_1,\ldots,S_t.$ Essentially, this enforces the structure of diminishing returns without hurting the approximation guarantees. %

\section{A Mechanism Design Framework}
\label{sec:sealed-bid-framework}

In this section, we develop a mechanism design framework that is capable of converting the state-of-the-art submodular optimization algorithms to feasible mechanisms for procurement auctions.

As a warm-up, we first show that the classic VCG framework~\citep{vickrey1961counterspeculation,clarke1971multipart,groves1973incentives} provides mechanisms that are IC, IR, and welfare-efficient. It turns out that the VCG mechanisms also satisfy NAS.

\begin{proposition}\label{prop:vcg-nas}
    The VCG mechanism satisfies NAS when $f$ is a submodular function.
\end{proposition}

The proof is postponed to \Cref{apx:sealed-bid-framweork}. Although VCG mechanisms are IC, IR, NAS, and welfare-efficient, 
implementing them is computationally prohibitive.
We now move on to describing
the computationally efficient
framework that transforms algorithms to mechanisms, which is one of the main contributions of our work.

Algorithm~\ref{alg:meta-algo}  provides a meta-algorithm $\A = (G)$ for regularized submodular function optimization specified by a scoring rule $G$, computing a score for a candidate $i$ given a subset $S$, a vector $\bm c$, the round number $k$, and possibly a random seed $r$ (for randomized algorithms) as input. The algorithm runs for $n$ rounds\footnote{Our framework can be extended to accommodate algorithms that stop early without running all the $n$ iterations. To simplify the exposition, we let the algorithm run for longer by adding extra {\em dummy} rounds.} and maintains a tentative solution set $S_k$ at the end of each round $k$. In each round $k$ it calls $G$ to compute a score for each candidate not in the tentative solution set $S_{k-1}$, and then it identifies the candidate $i^*$ with the highest score (where ties are broken lexicographically). If the highest score is positive, $S_{k} = S_{k-1} \cup \{i^*\}$; otherwise $S_{k} = S_{k-1}$.

\begin{algorithm}[h]
  \caption{A meta algorithm $\A = (G)$ for submodular optimization}
  \label{alg:meta-algo}
  \KwData{A set of seller $\N$, a cost profile $\bm c$ from sellers, and a random seed $r$}
  \KwResult{A subset of sellers to purchase services from}
  $S_0 = \emptyset$\;
  \For{$k$ from $1$ to $n$}{
    $i^* = \argmax_{i \not\in S_{k-1}} G(i, S_{k-1}, \bm c, k, r)$\;
    \If{$G(i^*, S_{k-1}, \bm c, k, r) > 0$} {
      $S_{k} = S_{k-1} \cup \{i^*\}$\;
    }
    \Else{
      $S_{k} = S_{k-1}$\;
    }
  }
  \Return{$S_n$}\;
\end{algorithm}

\begin{assumption}\label{assump:meta-algo}
    The meta-algorithm $\A = (G)$ satisfies 
    \begin{enumerate}
        \item for all $i$ and $S$ with $i \not\in S$, $G(i, S, \bm b, k, r)$ is non-increasing in $b_i$, for all $\bm b_{-i}$, $k$ and $r$;
        \item for all $i$ and $S$ with $i \not\in S$, if $b_i > f(i \mid S)$, then $G(i, S, \bm b, k, r) < 0$ for all $\bm b_{-i}$, $k$, and $r$;
        \item for all $i$, $S$ with $i \not\in S$, $k$ and $r$, $G(i, S, \bm b, k, r)$ is independent of $\bm b_{-i}$.
    \end{enumerate}
\end{assumption}

We argue that both (1) and (2) of Assumption~\ref{assump:meta-algo} are mild ones and any reasonable meta-algorithm $\A$ should satisfy it. For instance, all the algorithms
we present in \Cref{tab:mechanism-guarantees-different-submodular-algorithm}
satisfy these assumptions.

In particular, Assumption~\ref{assump:meta-algo}(1) states that the scoring function $G$ should be non-increasing as the reported bid $b_i$ increases, which is a natural requirement as a candidate with a smaller reported bid is more favorable. Assumption~\ref{assump:meta-algo}(2) states that the algorithm should not pick a candidate whose marginal contribution is smaller than their reported bid. Under truthful reporting, such a candidate has a negative marginal contribution in round $k$ towards the social welfare, and therefore, they should not be included to the solution. Assumption~\ref{assump:meta-algo}(3) is a stronger assumption, stating that the score for a seller $i$ should be independent of bids from other sellers, but to the best of our knowledge, almost all state-of-the-art algorithms satisfy this assumption. %
In fact, for our mechanism to satisfy the desired properties we can relax Assumption~\ref{assump:meta-algo}(3) to require: for all $i$, $S$ with $i \in S$, $k$, and $r$, for any $\bm b$, if $i \neq \argmax_{\ell \not\in S} G(i, S, \bm b, k, r)$, then for any $b'_i > b_i$, $\argmax_{\ell \not\in S} G\big(\ell, S, (b'_i, \bm b_{-i}), k, r\big) = \argmax_{\ell \not\in S} G(\ell, S, \bm b, k, r)$. In other words, as long as seller $i$ does not have the highest score, then the candidate with the highest score remains the same. Such a property is similar to the {\em non-bossiness} property recently studied by \citet{leme2023nonbossy}.

Given a meta algorithm $\A$ specified by Algorithm~\ref{alg:meta-algo}, Algorithm~\ref{alg:meta-algo-payment} first runs Algorithm \ref{alg:meta-algo} with the reported bids as input in order to obtain
the set of sellers $S^*$ whose items will be purchased.
Then, the payment for each seller $i \in S^*$ is computed in the 
following way: we re-run $\A$ by raising the bid from seller $i$ to infinity and record the intermediate solutions $\{S_0, S_1, \cdots, S_n\}$. For each $S_k$, we compute the supremum of the set of bids $b_i \geq 0$ satisfying $i = \argmax_{\ell \not\in S_k} G(\ell, S_k, \bm b, k, r)$ and $G(i, S_k, \bm b, k, r) > 0$. If such a non-negative bid does not exist, the sup function takes a default value of $0$. Finally, $p_i$ is computed by taking the max across $k \in [n]$.

\begin{algorithm}[h!]
  \caption{A feasible mechanism construction for a given meta algorithm $\A$}
  \label{alg:meta-algo-payment}
  \KwData{A set of sellers $\N$, a bid profile $\bm b$ from sellers, and a meta algorithm $\A$}
  \KwResult{A subset of sellers to purchase from and a vector of payment to sellers}
  Generate a random seed $r$ if needed or set $r = 0$\;
  $S^* = \A(\N, \bm b, r)$\;
  \For{$i \in S^*$}{
    Run $\A\big((\infty, \bm b_{-i}), r\big)$ and record $\{S_0, S_1, \cdots, S_n\}$\;
    $p_i = 0$\;
    \For{$k \in [n]$}{
        $p_i = \max(p_i, \sup\left\{b_i\big|G\big(i, S_{k-1}, \bm b, k, r\big) > 0\right\}$\;
        $p_i = \max(p_i, \sup\left\{b_i\big|i = \argmax_{\ell \not\in S_k} G\big(\ell, S_{k-1}, \bm b, k, r\big) \right\}$\;
    }
  }
  \Return{$S^*$ and $\bm p$}\;
\end{algorithm}

\begin{theorem}\label{thm:meta-algo-mech}
    For a meta-algorithm $\A = (G)$ satisfying Assumption~\ref{assump:meta-algo}, the mechanism constructed using Algorithm~\ref{alg:meta-algo-payment} is feasible. 
\end{theorem}

The proof is postponed to \Cref{apx:sealed-bid-framweork}. To show the IC and the IR properties, we make use of \citet{myerson1981optimal} together
with our carefully designed payment rule with Assumption~\ref{assump:meta-algo}(1) and \ref{assump:meta-algo}(3). It is more technically difficult is to establish the NAS property, where we make use of the submodularity property of $f$ and Assumption~\ref{assump:meta-algo}(2).
It is worth highlighting that Theorem~\ref{thm:meta-algo-mech} does not require the scoring function $G$ to have a diminishing-return structure, i.e., $G(i, S, \bm b, j, r) \geq G(i, T, \bm b, k, r)$ for all $S \subseteq T$ and $j \leq k$; and the distorted greedy algorithm does not satisfy such a structure.

With Theorem~\ref{thm:meta-algo-mech}, we establish a framework for converting a submodular optimization algorithm to a mechanism satisfying IC, IR, and NAS. In such a mechanism, the sellers are incentivized to submit their true private costs as their bids, and therefore, the mechanism preserves the bi-criteria welfare approximation guarantee, which follows immediately from the fact that under
truthful bidding the allocations of Algorithm \ref{alg:meta-algo} and Algorithm \ref{alg:meta-algo-payment} coincide.

\subsection{Online Mechanism Design Framework}
\label{sec:online-setting}

In this section we shift our attention to the 
online setting, where each seller $i \in \N$ arrives online
in an arbitrary order and the auctioneer needs to
make an irrevocable decision of whether to buy
the item or not. For convenience, assume seller $k$ arrives in round $k$. A meta-algorithm $\A^o = (G)$ for online submodular optimization is provided in Algorithm~\ref{alg:meta-algo-online}, where the scoring function $G$ in round $k$ computes a score for a seller $k$ given a subset $S$, a vector $\bm c_{(1,k)}$, and possibly a random seed $r$ as input. The algorithm maintains a tentative solution $S_k$ at the end of each round $k$. In each round $k$,  seller $k$ is added to the tentative solution if and only if the scoring function returns a positive score. 
From the taxation principle, we can focus on (possibly randomized) \emph{posted-price} mechanisms for designing IC and IR mechanisms, i.e.,
in each round $k$, the auctioneer makes a take-it-or-leave-it price $p_k \in \R_{\geq 0}$ for seller $k$. We provide an approximation-preserving transformation from online algorithms to posted-price mechanisms in Algorithm~\ref{alg:posted-price-mechanism}. Our main insight is to map the (online) scores to posted prices
in an approximation preserving way.
The formal proof is postponed to \Cref{apx:online-setting}.
Similarly as in the offline setting, we require that the online submodular optimization algorithm satisfies the following mild assumptions.

\begin{assumption}\label{assump:meta-algo-online}
    The meta-algorithm $\A^o = (G)$ satisfies 
    \begin{enumerate}
        \item for all $k$, $S$, $\bm b_{(1,k-1)}$ and $r$, $G(k, S, \bm b_{(1,k)}, r)$ is continuous and strictly decreasing in $b_k$;
        \item for all $k$, $S$, $\bm b_{(1,k-1)}$ and $r$, if $b_k > f(k \mid S)$, then \\
        $G(k, S, \bm b_{(1,k)}, r) < 0$.
    \end{enumerate}
\end{assumption}

In particular, we drop the counterpart of Assumption~\ref{assump:meta-algo}(3) in the online setting, but we additionally require $G$ to be continuous and strictly decreasing to avoid tie-breaking issues so that the equation $G\big(k, S, (\bm b_{(1,k-1)}, z), r\big) = 0$ has a unique solution in terms of $z$.

\begin{theorem}\label{thm:meta-algo-online-mech}
    For a meta-algorithm $\A^o = (G)$ satisfying Assumption~\ref{assump:meta-algo-online}, the online mechanism constructed using Algorithm~\ref{alg:posted-price-mechanism} is IC, IR, NAS,
    and outputs the same solution as Algorithm~\ref{alg:meta-algo-online} under full knowledge of the cost of the items. 
\end{theorem}

\begin{algorithm}[h]
  \caption{A meta algorithm $\A^o = (G)$ for online submodular optimization}
  \label{alg:meta-algo-online}
  \KwData{A set of sellers arriving online and a random seed $r$}
  \KwResult{A subset $S^*$ of sellers to purchase from}
  $S_0 = \emptyset, k = 0$\;
  \While{there exists a newly arrived seller $k+1$ with cost $c_{k+1}$}{
    $k = k + 1$\;
    \If{$G(k,S_{k-1}, \bm c_{(1,k)},r) > 0$}{
      $S_{k} = S_{k-1} \cup \{k\}$\;
    } \Else{
      $S_{k} = S_{k-1}$\;
    }
  }
  \Return{$S_k$}\;
\end{algorithm}

\begin{algorithm}[h]
  \caption{A posted-price mechanism construction for a given meta algorithm $\A^o$}
  \label{alg:posted-price-mechanism}
  \KwData{A set of sellers arriving online and a meta algorithm $\A^o$}
  \KwResult{A subset of sellers to purchase from and a vector of payment to sellers}
  Generate a random seed $r$ if needed or set $r = 0$\;
  $S_0 = \emptyset, k = 0$\;
  \While{there exists a newly arrived seller $k+1$}{
    $k = k + 1$\;
    Let $\hat p_k$ be the unique solution of the equation $G\Big(k, S_{k-1}, \big(\bm c_{(1,k-1)}, z\big), r\Big) = 0$ in terms of $z$\;
    Post price $\hat p_k$ to seller $k$\;
    \If{Seller $k$ accepts the posted price}{
      $S_k = S_{k-1} \cup \{k\}$\;
      $p_k = \hat p_k$\;
    }
    \Else{
      $S_k = S_{k-1}$\;
      $p_k = 0$\;
    }
  }
  \Return{$S_k$ and $\bm p$}\;
\end{algorithm}

As we alluded to before in Section~\ref{sec:submodular-algos}, the cost-scaled greedy algorithm \citep{nikolakaki2021efficient} 
gives a $(\frac12,1)$-bi-criteria approximation guarantee in the online setting
and fits within the template we have provided. 
Similarly, an adaptation of this 
algorithm by \citet{wang2021maximizing} gives
$(\nicefrac{\beta-\alpha}{\beta}, \nicefrac{\beta-\alpha}{\alpha}), 0 < \alpha \leq \beta$
parameterized bi-criteria guarantees in the 
online setting, and, interestingly, can 
also handle matroid constraints. In addition,
\citet{wang2020online} provide parameterized
bi-criteria guarantees adapting the algorithm
from \citet{nikolakaki2021efficient}.
In the
case of noisy evaluations of the function
$f$, our framework can be applied to the noise-robust algorithm
from \citet{gong2023algorithms}, which builds upon \citet{nikolakaki2021efficient}.

\section{Descending Auctions}\label{sec:descending-auctions}

The mechanism design framework developed in Section~\ref{sec:sealed-bid-framework} has a sealed-bid format. In this section, we investigate another popular class of procurement auctions, called {\em descending auctions}. A descending auction is parameterized by a demand oracle $\D$, which takes a set of sellers $S \subseteq \N$ and a vector of prices $\bm p$ as input and returns a subset of sellers $\D(S, \bm p) \subseteq S$. The auction maintains a set of active sellers $S$, initialized to be $\N$, and a vector of prices $\bm p$, where $p_i$ is initialized to be $f(i \mid \emptyset)$, the highest possible marginal contribution from seller $i$. In each iteration, the auction calls $\D$:
\begin{itemize}
    \item If $\D(S, \bm p) = S$, the descending auction ends and returns the current $S$ as the set of sellers to purchase from and $\bm p$ as the vector of payment to sellers
    \item If $\D(S, \bm p) \subsetneq S$, an active seller $i$ not in $\D(S, \bm p)$ is chosen and seller $i$'s price $p_i$ is decreased by a sufficiently small step size $\varepsilon$. If $p_i$ is smaller than $b_i$, then seller $i$ is removed from the active seller set and $p_i$ is set to $0$.
\end{itemize}

\begin{algorithm}[h]
  \caption{Descending auction with a demand oracle $\D$ and a step size of $\varepsilon$}
  \label{alg:descending-auctions}
  \KwData{A set of sellers $\N$ and a bid profile $\bm b$ from sellers}
  \KwResult{A subset $S^*$ of sellers to purchase from and a vector of payment to sellers}
  Set the set of active sellers as $S = \N$\;
  Set initial prices as $p_i = f(i \mid \emptyset)$\;
  \While{$\D(S, \bm p) \subsetneq S$}{
    Select an arbitrary seller $i \in S \setminus \D(S, \bm p)$\;
    $p_i = p_i - \varepsilon$\;
    \If{$p_i < b_i$} {
        $S = S \setminus \{i\}$\;
        $p_i = 0$\;
    }
  }
  \Return{$S$ and $\bm p$}\;
\end{algorithm}
Algorithm~\ref{alg:descending-auctions} provides a classic paradigm of such auctions. Note that the descending auction is always IR since seller $i$ is chosen only if the final price $p_i$ is at least $b_i$, and IC since the reported bid $b_i$ is only used for checking whether the current price $p_i$ is smaller than $b_i$, which also enables the implementation without eliciting sealed bids from sellers in which whenever $p_i$ is lowered, the auctioneer asks the seller whether they would like to leave the market.
This has the appealing property that 
for the set of the winning sellers, the auctioneer only learns
an upper bound on their true valuation, instead of their actual
valuation, which is always the case in sealed-bid auctions. 
Another advantage of descending auctions is that they satisfy \emph{obviously strategyproofness} \citep{li2017obviously}. 
Moreover, NAS can be achieved if $f\big(\D(S, \bm p)\big) - \sum_{i \in \D(S, \bm p)} p_i \geq 0$ always holds. 

Descending auctions were initially designed under the assumption that $f$ is a gross-substitute function, a subclass of submodular functions. In particular, when $f$ is a gross-substitute function, the demand oracle that exactly solves the welfare optimization problem can be computed in polynomial time; and with such an oracle, the descending auction returns the optimal subset $\OPT$, even if it is executed in {\em the adversarial setting}, i.e., where the selection of seller $i$ not in $\D(S, p)$ for price decrement is determined by an adversary~\citep{kelso1982job,leme2017gross}. 

\noindent\textbf{Descending Auctions in the Adversarial Setting.}\label{sec:descending-adversarial}
We first show that when $f$ is a submodular function, in the adversarial setting, a descending auction may return an arbitrarily worse solution even with the exact demand oracle satisfying bi-criteria $(1, 1)$-approximation in welfare. The proof 
is postponed to \Cref{apx:descending-auctions}.

\begin{theorem}\label{thm:descend-negative}
    In the adversarial setting, given demand oracle $\D$ satisfying bi-criteria $(1, 1)$-approximation in welfare, for any $L \in \mathbb N_+$, there exists a problem instance with $L+2$ sellers and a vector of bids $\bm b$, such that Algorithm~\ref{alg:descending-auctions} with demand oracle $\D$ and $\varepsilon < \frac 1L$ returns a subset $S^*$ with $f(S^*) - \sum_{i \in S^*} b_i = O\left(\frac{1}{L}\right) \cdot \big(f(\OPT) - \sum_{i \in \OPT} b_i\big)$.
\end{theorem}

We complement our negative result by designing
a demand oracle based on the cost-scaled greedy algorithm of \cite{nikolakaki2021efficient}, which leads to a descending auction with $(\nicefrac{1}{2},1)$-approximation
guarantees. The proof is postponed to \Cref{apx:descending-auctions}.

\begin{theorem}\label{thm:descending-auction-adversarial}
    There exists a demand oracle $\hat \D$ satisfying bi-criteria $(\frac 12, 1)$-approximation in welfare such that, Algorithm~\ref{alg:descending-auctions} with demand oracle $\hat \D$ and $\varepsilon > 0$ always returns a subset $S^*$ satisfying $f(S^*) - \sum_{i \in S^*} b_i \geq \frac 12 f(\OPT) - \sum_{i \in \OPT} b_i - n \varepsilon$ in the adversarial setting.
\end{theorem}

\begin{remark}\label{rem:descending-auctions}
    We remark that our results show a, perhaps,
    counter-intuitive phenomenon in the adversarial setting: there are instances
    in which if we run the descending auction described
    in \Cref{alg:descending-auctions} with the
    \emph{perfect} demand oracle we will end up
    with welfare that is arbitrarily worse 
    than the execution with the oracle based on \citet{nikolakaki2021efficient}. In particular,
    the family of problems that witnesses the lower
    bound in \Cref{thm:descending-auction-adversarial}
    shows that
    for every $L > 0,$
    there is an instance in which the solution
    $S^*$ that we get through the perfect oracle
    satisfies $f(S^*) - \sum_{i \in S^*}b_i < 1,$
    whereas the solution $\hat S$ that we get through
    the scaled-greedy oracle satisfies $f(\hat S) - \sum_{i \in \hat S}b_i \geq L/2 - 1.$
\end{remark}

\noindent\textbf{From Online Submodular Optimization to Descending Auctions.}\label{sec:online-to-descending} We demonstrate a reduction from online submodular optimization to descending auctions when the selection of seller $i$ not in $\D(S, p)$ for price decrement can be controlled by the auctioneer. Recall that we have demonstrated that an online submodular optimization algorithm (Algorithm~\ref{alg:meta-algo-online}) that satisfies Assumption~\ref{assump:meta-algo-online} can be converted to an (online) posted-price auction that preserves the bi-criteria welfare guarantee. As posted-price auctions can be implemented by descending auctions with a tailored schedule of descending prices, an online submodular optimization algorithm (Algorithm~\ref{alg:meta-algo-online}) that satisfies Assumption~\ref{assump:meta-algo-online} can also be converted to a descending auction that preserves the bi-criteria welfare guarantee (see Algorithm~\ref{alg:online-algo-to-descending-auctions} in \cref{apx:descending-auctions}).
Thus, an impossibility result, showing there is no descending auction that can achieve bi-criteria $(\alpha,1)$-approximation in welfare with $\alpha > \frac 1 2$, directly implies an impossibility result on online submodular optimization with bi-criteria $(\alpha,1)$-approximation guarantees for the same $\alpha$, which is a long-standing open question for online submodular optimization. We leave this as an interesting open question for future research.

\section{Experiments} \label{sec:experiment} 
In this section, we empirically evaluate the welfare performance and running time trade-offs of various mechanisms on a publicly available coverage problem.
Our instances are constructed using a bipartite graph from SNAP (\url{https://snap.stanford.edu/data/wiki-Vote.html}). The graph consists of approximately 7000 nodes representing {\em sets} and 2800 nodes representing {\em vertices} to be covered by the sets.
We consider the value of covering each vertex and the cost of selecting each set based on the degrees of the corresponding nodes.
We define the value of $f(S)$ as the sum of the values from vertices covered by at least one selected set, i.e., $f(S) = \sum_{i \in \bigcup_{s \in \mathcal S} s} v_i$. To create instances with varying sizes and value-to-cost ratios, for each pair $(n, s)$, we randomly sample subsets of the sets of size $n$ and scale their costs using $\kappa \sim U[s, s^2]$. We generate $10000$ random instances per $(n, s)$.

Since VCG needs to solve exactly 
$\arg\max_S f(S) - \sum_{i\in S}c_i$
we solve the optimization problem using a mixed integer program (MIP), when
$n$ is relatively small.
\Cref{fig:runtime} shows the running time comparison of different methods, while welfare results are shown in \Cref{fig:welfare-small} and \Cref{fig:welfare-large}.
In summary: 
\begin{itemize}
    \item VCG computes both VCG allocation and payment;
    \item Optimal Welfare computes VCG allocation (i.e., the welfare-optimal allocation) only; 
    \item For Descending Auction (DA) with an unspecified oracle, we implement it with a demand oracle that computes the welfare-optimal allocation. 
\end{itemize}
In addition, we experiment with Greedy-margin, Greedy-rate, Distorted Greedy, and Cost-scaled Greedy, introduced in \Cref{sec:submodular-algos} and their DA variants.
In \Cref{apx:experiments}, we explain a heuristic approach, based on lazy implementations of the classical greedy algorithm \citep{minoux2005accelerated}, to speed up the computation for both allocation and payment.

\begin{figure}[htbp]
    \centering
    \includegraphics[width=0.5\textwidth]{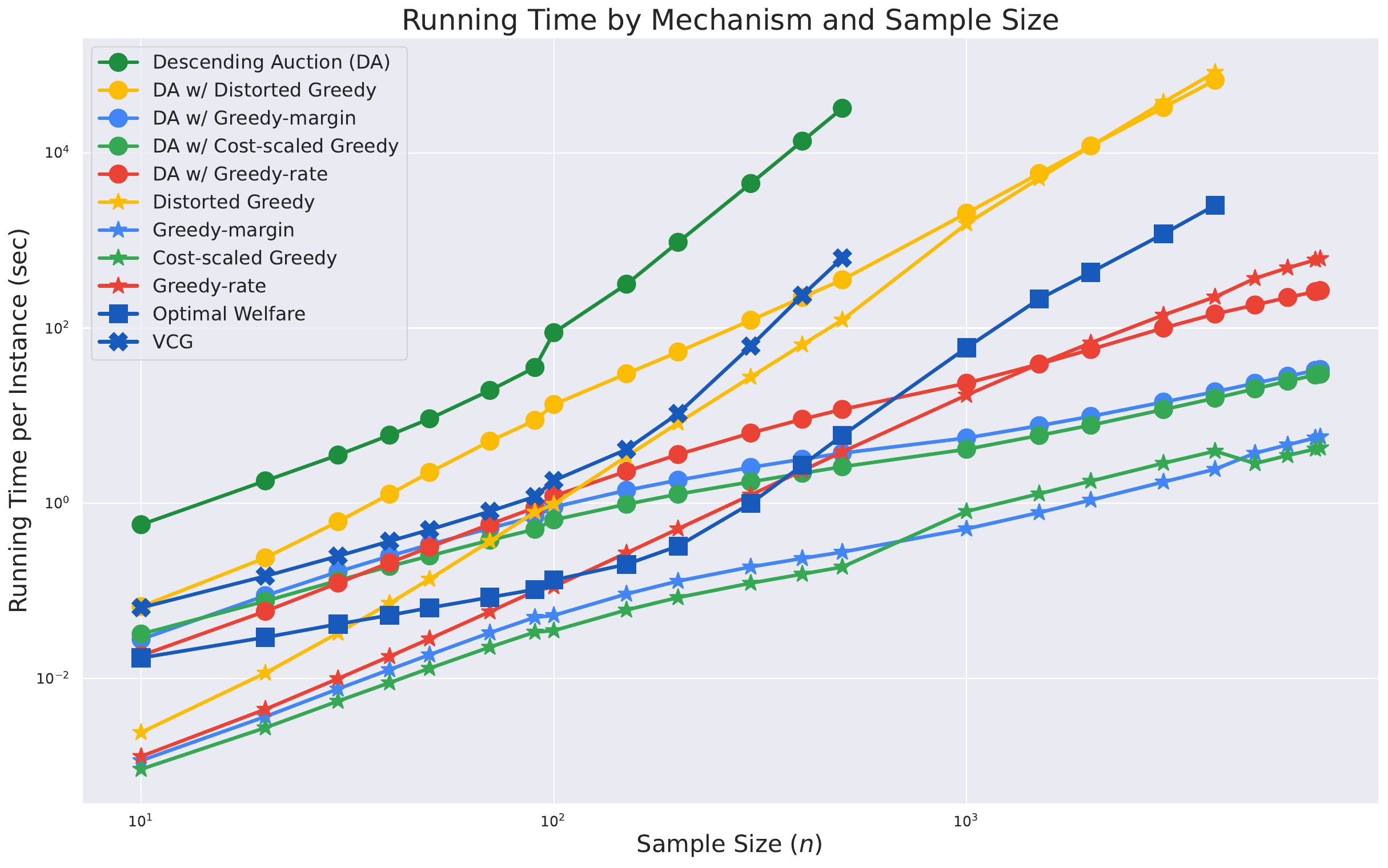}
    \caption{Average running time for different mechanisms at different sample sizes \(n\). Note that both axes are log-scaled.}
    \label{fig:runtime}
\end{figure}

\paragraph{Running Time Comparison}
\Cref{fig:runtime} shows the average running time for different mechanisms. VCG, Descending Auction, and Optimal Welfare exhibit super-polynomial complexity due to MIP computations, with Descending Auction being the slowest. Greedy-based algorithms demonstrate polynomial complexity, but Distorted Greedy's time grows faster due to the absence of a diminishing-return structure, so the approach from \Cref{apx:experiments} does not apply.

\paragraph{Welfare Comparison}
\Cref{fig:welfare-small} and \Cref{fig:welfare-large} compare welfare of various mechanisms for varying sample sizes. Instances are grouped by the fraction of active agents, determined by whether their marginal contribution exceeds their cost. In particular, for valuation function \(f\) and cost vector \(\bm c\), the fraction of active agents is \(|\{i \in \N \mid f(i \mid \emptyset) > c_i\}| / |\N|\), which decreases as the cost multiplier \(\kappa\) increases. When feasible to run, DA with an optimal oracle outperforms other auctions. This contrasts our theoretical result, which shows that in pathological instances, DA with an optimal oracle can be arbitrarily worse than cost-scaled greedy. Across approximation algorithms, we observe that the direct implementations consistently outperform their DA variants, with an ordering of Greedy-margin > Greedy-rate > Cost-scaled Greedy > Distorted Greedy.

\begin{figure}[htbp]
    \centering
    \begin{subfigure}[b]{0.48\textwidth}
        \includegraphics[width=\textwidth]{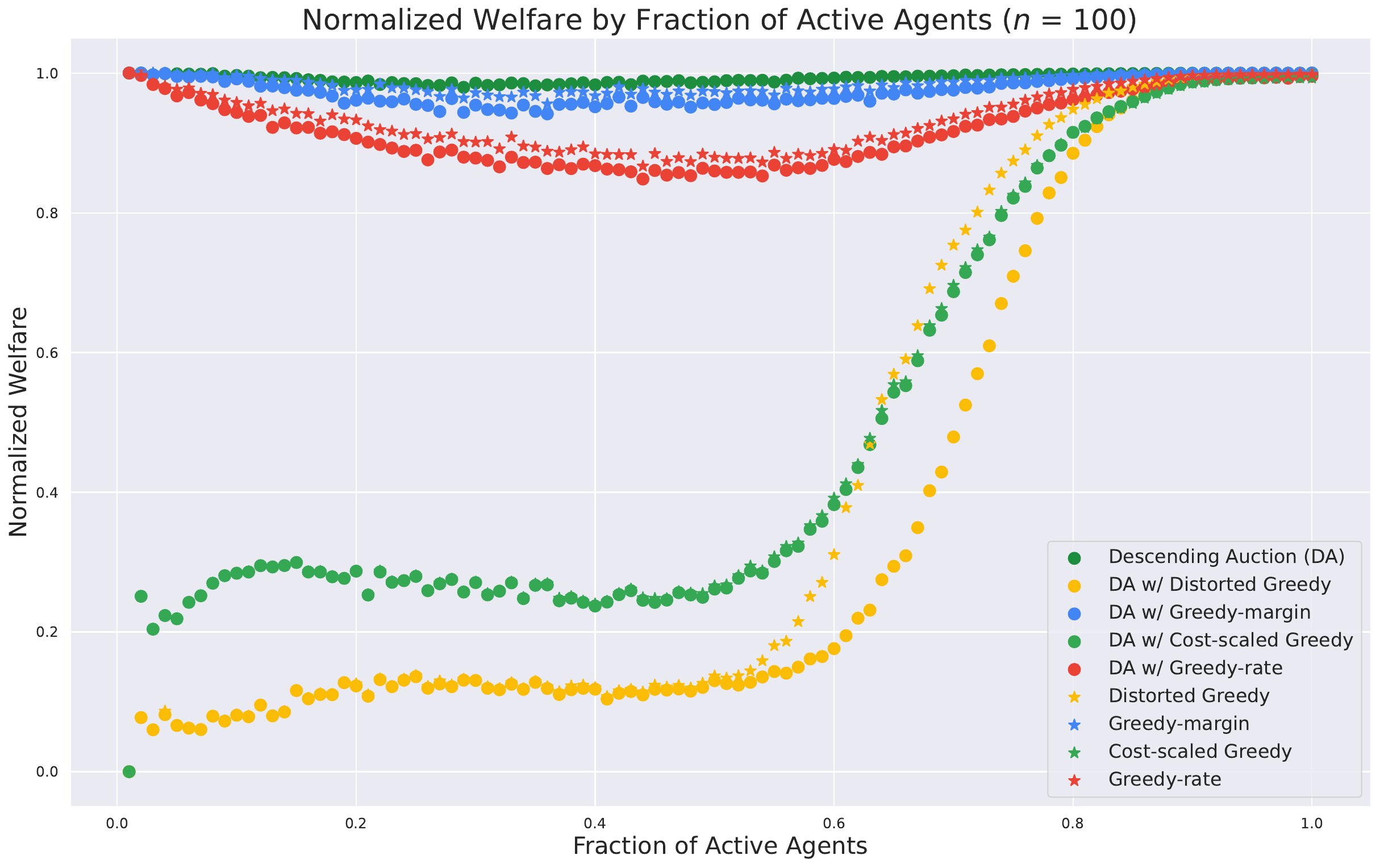}
        \caption{\(n = 100\)}
        \label{fig:welfare-small-n100}
    \end{subfigure}
    \begin{subfigure}[b]{0.48\textwidth}
        \includegraphics[width=\textwidth]{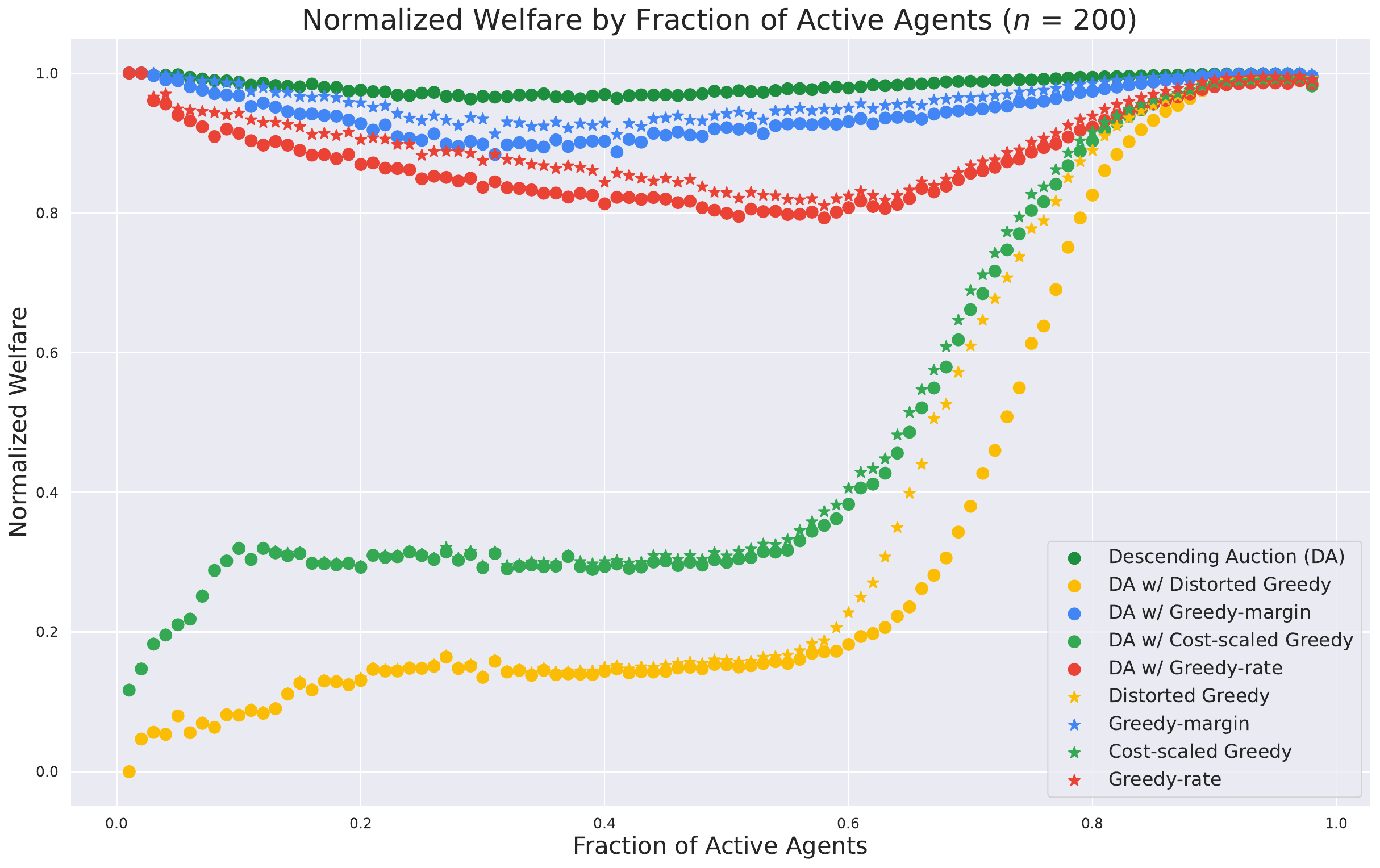}
        \caption{\(n = 200\)}
        \label{fig:welfare-small-n200}
    \end{subfigure}
    \begin{subfigure}[b]{0.48\textwidth}
        \includegraphics[width=\textwidth]{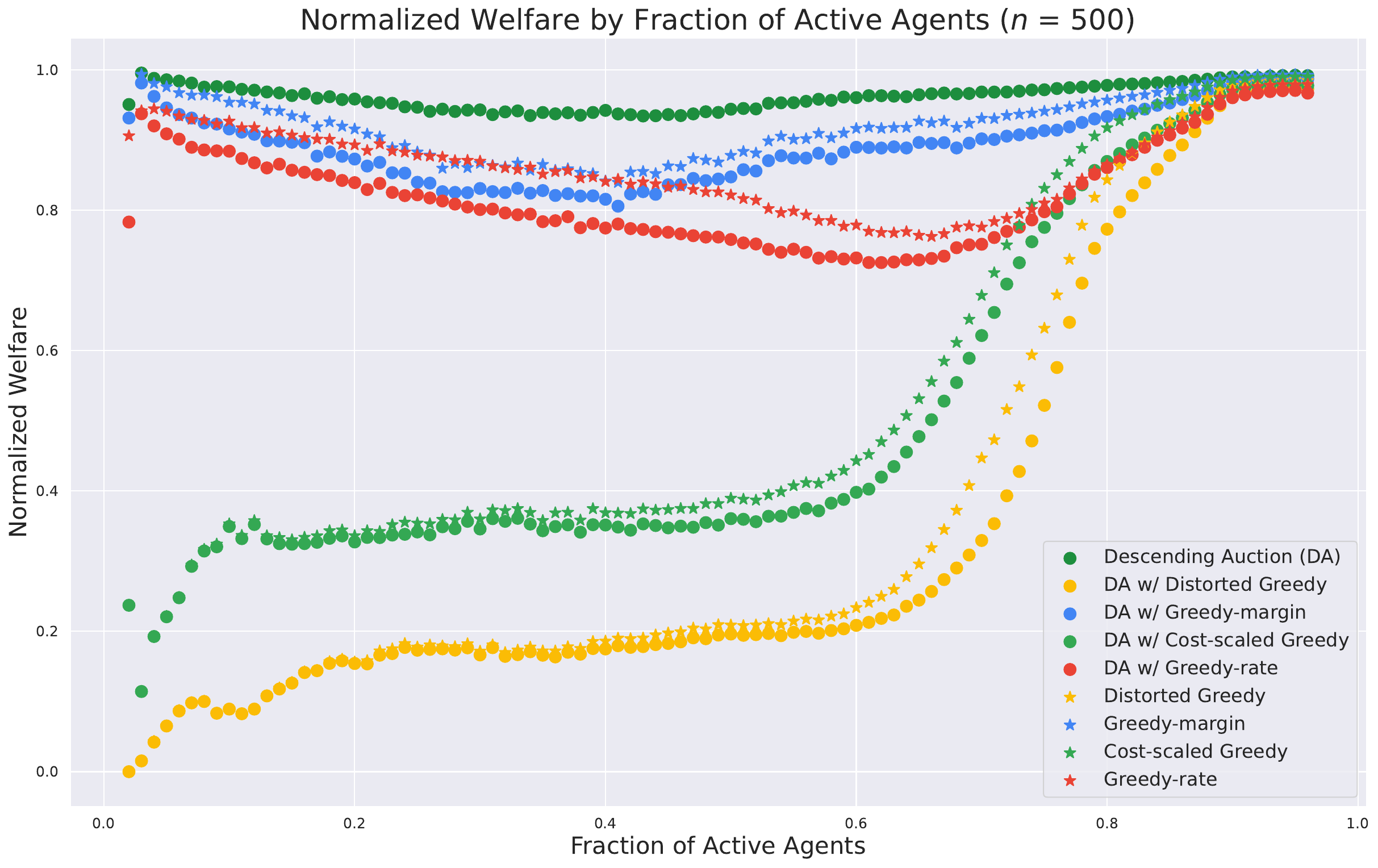}
        \caption{\(n = 500\)}
        \label{fig:welfare-small-n500}
    \end{subfigure}
    \caption{Welfare as a function of the fraction of active agents for \(n \in \{100, 200, 500\}\).}
    \label{fig:welfare-small}
\end{figure}

\begin{figure}[htbp]
    \centering
    \begin{subfigure}[b]{0.48\textwidth}
        \includegraphics[width=\textwidth]{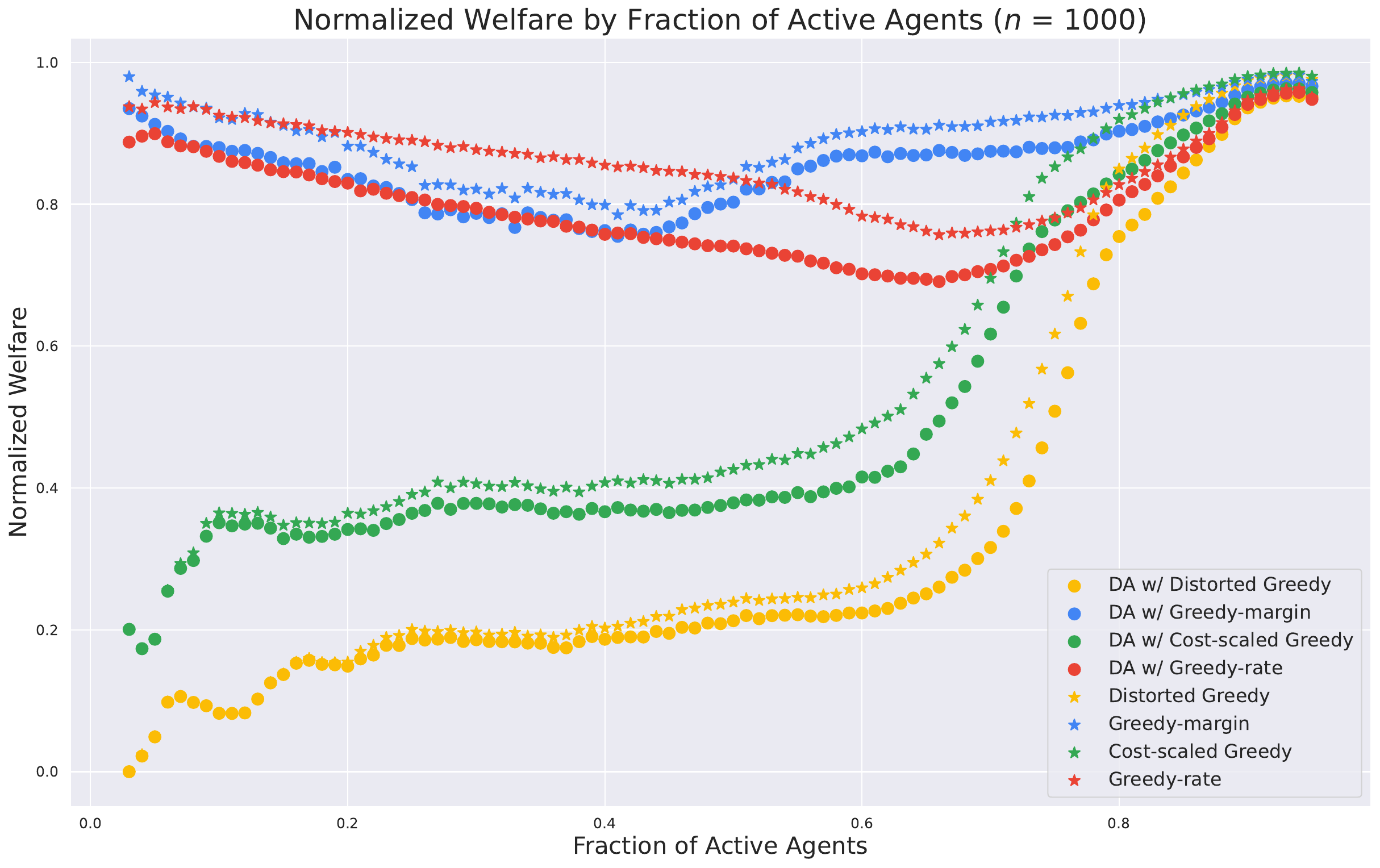}
        \caption{\(n = 1000\)}
        \label{fig:welfare-large-n1000}
    \end{subfigure}
    \begin{subfigure}[b]{0.48\textwidth}
        \includegraphics[width=\textwidth]{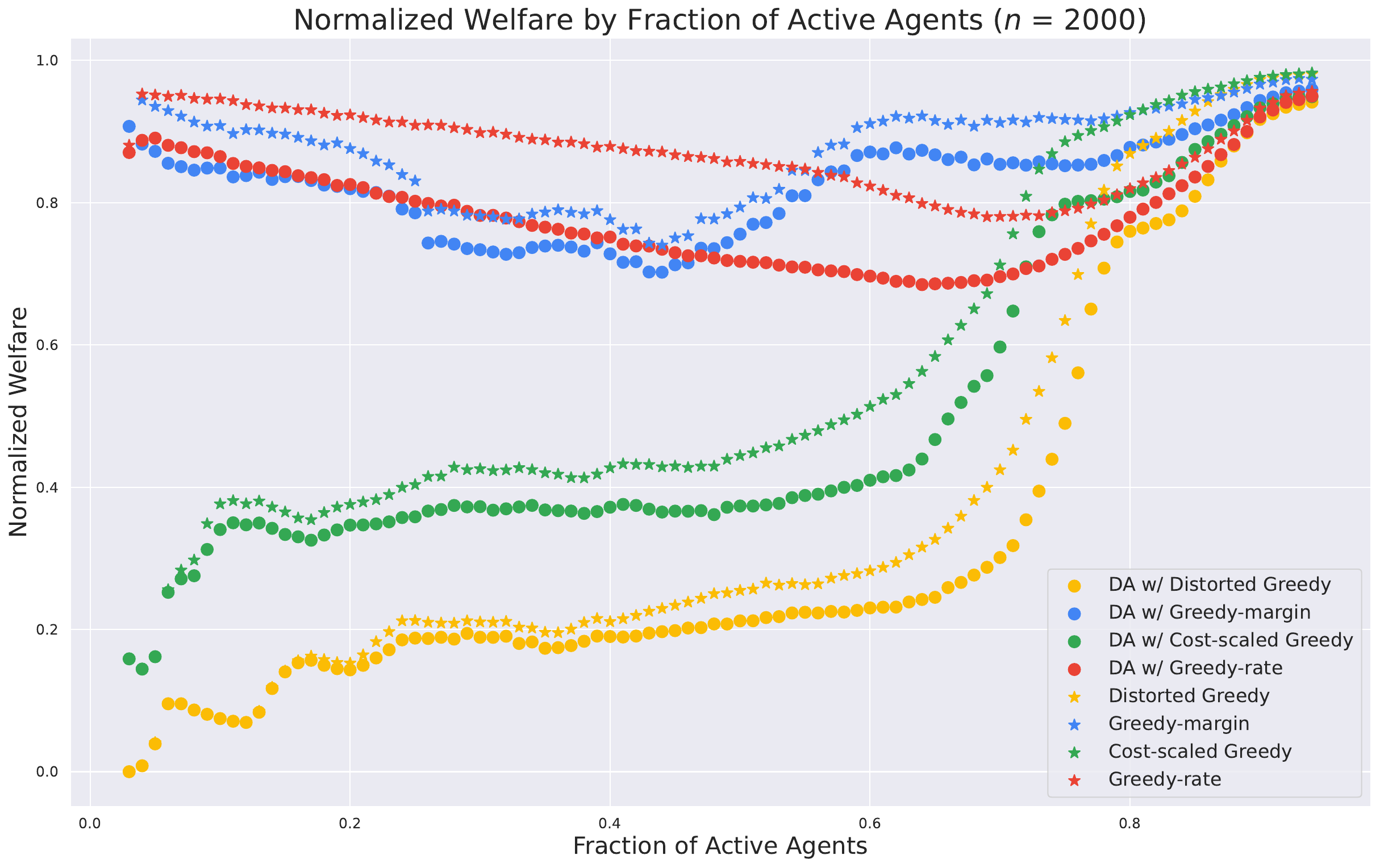}
        \caption{\(n = 2000\)}
        \label{fig:welfare-large-n2000}
    \end{subfigure}
    \begin{subfigure}[b]{0.48\textwidth}
        \includegraphics[width=\textwidth]{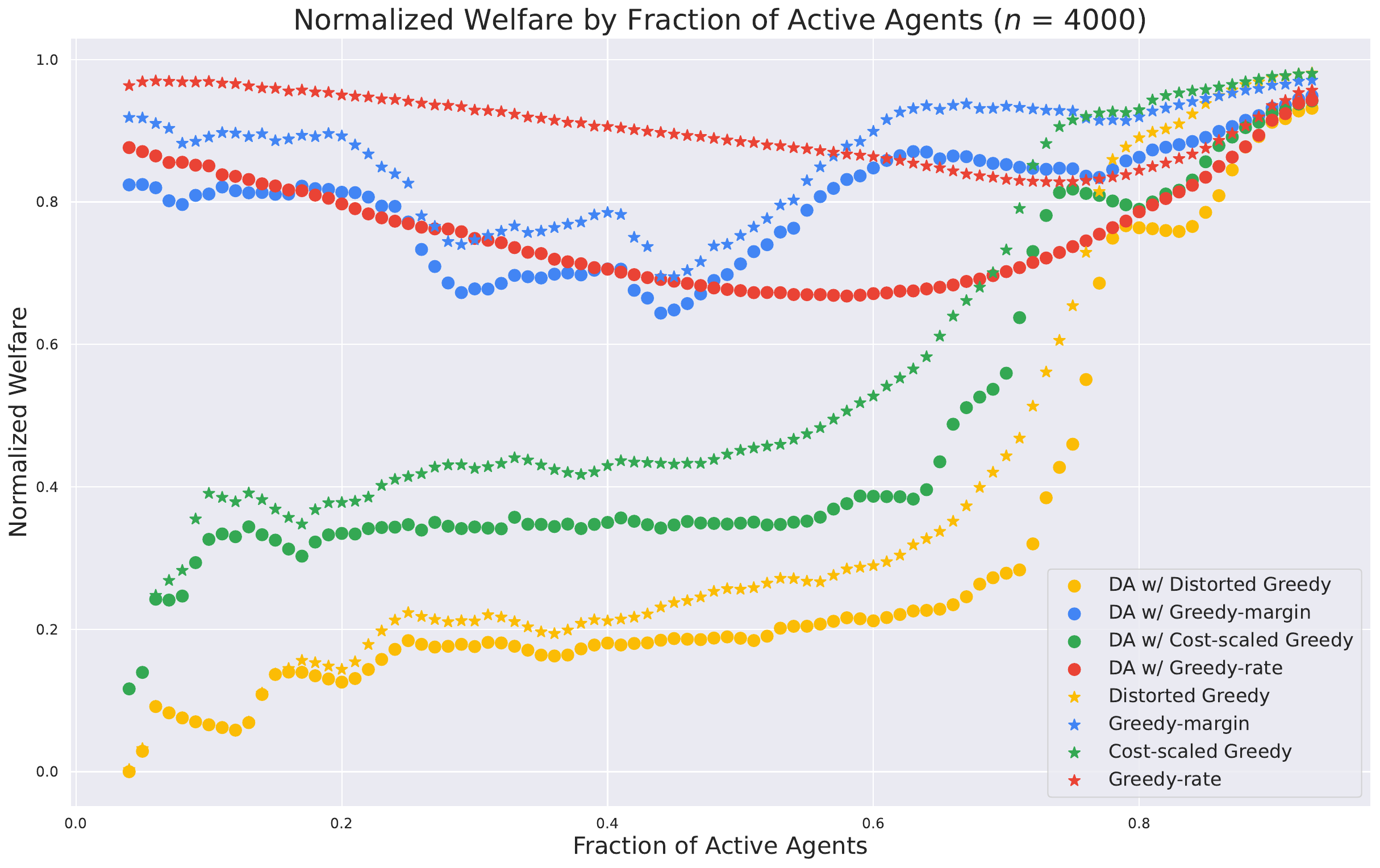}
        \caption{\(n = 4000\)}
        \label{fig:welfare-large-n4000}
    \end{subfigure}
    \caption{Welfare as a function of the fraction of active agents for \(n \in \{1000, 2000, 4000\}\).}
    \label{fig:welfare-large}
\end{figure}

\section{Conclusion}\label{sec:conclusion}
In this work we propose
a new procurement auction
setting inspired by the recent development in regularized submodular optimizations. Our results
enable computational and welfare
efficient transformations from regularized submodular
maximization algorithms to various
types of mechanisms, including sealed-bid auctions and descending auctions,
that satisfy several desirable
properties such as NAS, IC, and IR.
Moreover, we have experimentally 
tested our framework
on several large-scale instances
which showcases its practical
applicability.

\clearpage

\bibliographystyle{plainnat} %
\bibliography{references} %

\clearpage
\onecolumn
\appendix

\section{Related Work}\label{apx:related-work}
\paragraph{Gains From Trades vs. Utilitarian Objectives.} 
The objective we study in this work has connections
to gains-from-trade in two-sided market.
The seminal work of~\citet{myerson1983efficient}
showed that even when there is one 
seller and one buyer
there is no mechanism
that satisfies IR, Bayesian IC (BIC), Budget-Balance,
and can extract the full gains-from-trade (GFT), i.e.,
the total value generated by transferring the items
from the sellers to the buyers. In light of that result, a lot of
works have focused on relaxing one of these desiderata. Of particular 
interest are the works that are seeking an \emph{approximately} optimal
solution to the GFT objective. In the case of one buyer and one seller, \citet{mcafee1992dominant} provided an elegant mechanism that achieves
an $1/2$-approximation guarantee if the median of the buyer's value
is higher than that of the seller's. In a similar spirit,
\citet{blumrosen2016approximating} proposed a $1/e-$approximation
for this problem, when both the buyer and the seller satisfy the monotone-hazard-rate (MHR) condition. Recently, \citet{brustle2017approximating}
provided a simple mechanism that gives a $1/2-$approximation to GFT
under arbitrary distributions. Subsequently, \citet{deng2022approximately}
provided a constant approximation to the \emph{first-best} objective, i.e.,
the GFT that can be achieved if the agents are not strategic, resolving 
a long-standing open question; and \citet{fei2022improved} improves the approximation factor to $3.15$. \citet{cai2021multi} moved beyond
the setting with one seller, providing approximations
algorithms in environments with multiple sellers.
A different line of work in the two-sided markets literature has
focused on the utilitarian objective,
i.e., the total 
value of the buyers
for the items that the purchased and the total cost of
the sellers who did not sell their items. 
\citet{lehmann2001combinatorial} gave a greedy algorithm
for this objective that gives a $1/2-$approximation,
and \citet{fu2012conditional} implemented it as an ascending auction
that gives the same approximation, when the buyers are not strategic. More recently, 
\citet{blumrosen2021almost}
designed a simple take-it-or-leave-it mechanism 
for this objective
that achieves a $(1-1/e)-$approximation. \citet{colini2016approximately} designed a constant-factor
approximation mechanism that satisfies the Strongly Budget Balanced condition, meaning that all the payments
of the buyers are transferred to the sellers.
Subsequently, \citet{colini2020approximately} developed
a constant-factor approximation mechanism in the
setting with XOS valuations.

\paragraph{Budget-feasible Mechanism Design.}
The study of budget-feasible mechanism design, with an 
emphasis on procurement auctions, was put forth by the seminal
work on \citet{singer2010budget} who provided 
a prior-free\footnote{This means that the auctioneer does not 
have any prior information about the private types of the sellers.}
budget-feasible mechanism that enjoys a constant approximation
when the objective is a \emph{non-negative} and monotone
submodular function. Since then, a long line of work
has studied this problem providing better approximation guarantees, relaxing the monotonicity
constraint on the objective function and extending
the results to even more general classes of functions
including XOS and subadditive functions \citep{dobzinski2011mechanisms,chen2011approximability, anari2014mechanism,bei2017worst,jalaly2018simple,amanatidis2019budget,balkanski2022deterministic,han2024triple}. Currently, the best approximation
guarantees for monotone submodular functions
are obtained by \citet{han2024triple}, that achieve
an approximation factor of 4.3. It is worth highlighting
that, to the best of our knowledge, none of these works have studied
\emph{non-positive} submodular functions.

\paragraph{Reduction from Algorithm Design to Mechanism Design.} Our
results contribute to the rich literature
of transformations from \emph{algorithms} that
operate on non-strategic data to \emph{mechanisms}
whose input is coming from strategic agents.
A striking result by \citet{chawla2012limits}
showed that there is no welfare-preserving
black-box reduction when the mechanism
is required to be truthful in expectation.
Since then, a beautiful line of work initiated by \citet{HartlineL10}
develops mechanisms that satisfy the weaker Bayesian Incentive Compatibility
(BIC) condition, instead of truthfulness in expectation. 
An important tool that all these works \citep{HartlineL10,BeiH11,HartlineKM11,hartline2015bayesian,dughmi2017bernoulli} utilize
to establish the (approximate) BIC property is a \emph{replica-to-surrogate} matching. It is worth highlighting that \citet{dughmi2017bernoulli} managed
to obtain an \emph{exactly} BIC mechanism, by introducing novel constructions
in the context of mechanism design, such as various \emph{Bernoulli factories.}
In the context of multidimensional revenue maximization, \citet{cai2012optimal,cai2012algorithmic,cai2013reducing,cai2013understanding}
developed black-box transformations from algorithms to
(approximately) revenue-optimal and (approximately) BIC mechanisms

\paragraph{Submodular Optimization.} 
The problem of submodular maximization has received a lot 
of attention in the optimization literature.
The seminal work of 
\citet{nemhauser1978analysis} shows that the natural
greedy algorithm gives an $(1-1/e)-$approximation to the 
optimal solution when the submodular function
is \emph{monotone} and \emph{non-negative}. Later, \citet{feige2011maximizing}
designed algorithms that provide constant approximation guarantees
when the underlying function is \emph{non-negative}, but, potentially, non-monotone. Designing approximation
algorithms with constant-factor approximation guarantees 
for general non-positive submodular functions 
is known to be hard \citep{papadimitriou1988optimization,feige1998threshold}.
The study of the objective $f(S) - c(S),$ 
where
$f$ is a non-negative submodular function and $c$
is a linear function,
is commonly referred to as \emph{regularized submodular optimization}.
This line of work was initiated by \citet{sviridenko2017optimal} who designed
an algorithm that obtains a $(1-1/e)f(\mathrm{OPT}) - c(\mathrm{OPT})$ and showed
that this guarantee is optimal. Later, \citet{feldman2021guess} and \citet{harshaw2019submodular} designed simpler
and more computationally efficient algorithms that attain
the same approximation guarantees. Since then, there has been a long
line of work studying that problem and designing algorithms
that work in the centralized setting, the distributed setting, the
streaming setting, and the online setting \citep{kazemi2021regularized,wang2020online,wang2021maximizing,jin2021unconstrained,mitra2021submodular+,nikolakaki2021efficient,gong2021multi,tang2021adaptive,tang2021submodular,geng2022bicriteria,lu2023regularized,qi2023maximizing,lu2023streaming,gong2023algorithms}.
It is worth noting that the design of approximation algorithms
for bi-criteria objectives has alson been considered in different
contexts in the past \citep{kleinberg2004segmentation,feige2008combinatorial,feige2013pass}.

\section{Omitted Details from \Cref{sec:submodular-algos}}\label{apx:regularized-submodular-optimization}

\begin{table}[ht]
\centering
\caption{Instantiations of Algorithm \ref{alg:meta-algo-payment} with different
submodular optimization algorithms.}
\label{tab:mechanism-guarantees-different-submodular-algorithm}
\begin{tabular}{ p{6cm} p{7cm} }
    \toprule
    
    \hspace{19mm}\small{\underline{Algorithm}}
    & \hspace{1.8mm}\small{\underline{Approximation Guarantee}} \\
    \midrule
Greedy-margin & No Worst-Case Guarantee \\
\citep{kleinberg2004segmentation} & \\
\midrule
Greedy-rate & Parametrized Guarantee \\
\citep{feige2013pass} & \\
    \midrule
Distorted Greedy & $(1-e^{-\beta})\cdot f(\mathrm{OPT}) - (\beta + o(1)) \cdot c(\mathrm{OPT)}$\\
\citep{harshaw2019submodular} \& (this paper) & for all $\beta \in [0,1]$\\
\midrule
Stochastic Distorted Greedy & $(1-e^{-\beta})\cdot f(\mathrm{OPT}) - (\beta + o(1)) \cdot c(\mathrm{OPT)}$  \\ 
\citep{harshaw2019submodular} \& (this paper) & for all $\beta \in [0,1]$ (in expectation) \\
\midrule
ROI Greedy (same as Greedy-rate) & $f(\mathrm{OPT}) - \left(1+\ln\frac{f(\mathrm{OPT})}{c(\mathrm{OPT})}\right) \cdot c(\mathrm{OPT})$ \\
\citep{jin2021unconstrained} & \\
    \midrule
Cost-scaled Greedy & $1/2 \cdot f(\mathrm{OPT}) - c(\mathrm{OPT})$\\ 
\citep{nikolakaki2021efficient} & (online and adversarial) \\
    \bottomrule

\end{tabular}
\end{table}

\subsection{Deterministic Distorted Greedy Algorithm Analysis}
\begin{algorithm}[h]
  \caption{Distorted greedy algorithm \citep{harshaw2019submodular}}
  \label{alg:distorted-greedy}
  \KwData{A set of $n$ items $\N$, a monotone submodular function $f:2^\N \rightarrow \R_{\geq 0}$, a cost function $c:\N \rightarrow \R_{\geq 0},$ a capacity
  constraint $k \leq n$}
  \KwResult{A subset of the items $R \subseteq \N$ with $\abs{S} \leq k$}
  $S_0 = \emptyset$\;
  \For{$j$ from $1$ to $k$}{
    $G(i, S_{j-1}, \bm c, j, r) = \left(1-\frac{1}{n}\right)^{k - j}\cdot f(i \mid S_{j-1}) - c _i, \forall i \in S$\;
    $i_t = \argmax_{i \not\in S_{j-1}} G(i, S_{j-1}, \bm c, j, r)$\;
    \If{$G(i^*, S_{j-1}, \bm c, j, r) > 0$} {
      $S_{j} = S_{j-1} \cup \{i^*\}$\;
    }
    \Else{
      $S_{j} = S_{j-1}$\;
    }
  }
  \Return{$S_n$}\;
\end{algorithm}

In this section we extend the results
of \cite{harshaw2019submodular} that obtain
a $(1-e^{-1},1)$ bi-criteria approximation guarantee
to $(1-e^{-\beta}, \beta)$ guarantee,
where $\beta \in [0,1].$ First, we 
recall a negative
result from \cite{feldman2021guess}.
\begin{theorem}[\cite{feldman2021guess}]\label{thm:bi-criteria-pareto-lower-bound}
    For every $\beta \in [0,1], \varepsilon > 0,$
    no polynomial-time algorithm can guarantee
    $(1-e^{-\beta}+\varepsilon,\beta)$-bi-criteria approximation
    for the problem $\max_S (f(S) - \sum_{i \in S}c_i)$.
\end{theorem}
We underline that \citet{feldman2021guess} proposes
an algorithm that achieves an (almost) matching upper bound,
however the algorithm
works using the continuous multilinear
extension of the underlying function, which
is quite non-practical and does not fit
our mechanism design framework. We propose an approach
based on the distorted greedy algorithm 
of \cite{harshaw2019submodular} that
achieves this guarantee for every
$\beta \in [0,1].$ Formally, we prove the result stated in \Cref{thm:bi-criteria-distorted}.

We reiterate that \Cref{thm:bi-criteria-distorted} 
improves upon the guarantees stated in \cite{harshaw2019submodular}
since it achieves optimal bi-criteria guarantee simulataneously for all
$\beta \in [0,1].$
Following \cite{harshaw2019submodular}, we define the following
functions
\begin{align*}
    \Phi_t(T) &= \left( 1- \frac{1}{k}\right)^{k-t}f(T) -   \sum_{j \in T} c_j, &\forall T \subseteq \N,\\
    \Psi_t(T,i) &= \max\left\{0,  \left( 1- \frac{1}{k}\right)^{k-(t+1)}f(i|T) -  c_i \right\}, & \forall T\subseteq \N, i \in \N.
\end{align*}

The next result from \cite{harshaw2019submodular} describes the connection between these two quantities.

\begin{lemma}[\cite{harshaw2019submodular}]\label{lem:connection-psi-phi}
    In each iteration $t \in \{0,\ldots,k-1\}$ of \Cref{alg:distorted-greedy} 
    it holds that
    \[
        \Phi_{t+1}(S_{t+1}) - \Phi_{t}(S_{t}) = \Psi_t(S_t,i_t) + \frac{1}{k}\left(1- \frac{1}{k}\right)^{k-(t+1)}f(S_t) \,.
    \]
\end{lemma}

The next result of \cite{harshaw2019submodular} relates the marginal gain in each iteration
of \Cref{alg:distorted-greedy} to the function $\Psi_t(S_t,i_t).$

\begin{lemma}[\cite{harshaw2019submodular}]\label{lem:psi-lower-bound}
    In each iteration of \Cref{alg:distorted-greedy} it holds that
    \[
        \Psi_t(S_t,i_t) \geq \frac{1}{k}\left(1-\frac{1}{k}\right)^{k-(t+1)}\left(f(\OPT) - f(S_t)\right) -\frac{1}{k} c(\OPT)
    \]
\end{lemma}

Finally, we will make use of the following simple technical result.

\begin{proposition}\label{lem:geom-series-exp-bound}
    For any $n \in \mathbb{N}, \beta \in [0,1]$ it holds that
    \[
        \frac{1}{ n}\cdot \frac{1-\left(1-\frac{1}{n}\right)^{x \cdot n}}{1-\left(1-\frac{1}{ n}\right)} \geq 1-e^{-x} \,.
    \]
\end{proposition}
\begin{proof}
    We have that
    \begin{align*}
         \frac{1}{ n}\cdot \frac{1-\left(1-\frac{1}{n}\right)^{x \cdot n}}{1-\left(1-\frac{1}{ n}\right)} &= 1-\left(1-\frac{1}{n}\right)^{x\cdot n} & (\text{Simplify denominator)}\\
         &\geq 1 - e^{-\frac{x\cdot n}{n}}. &(1-x \leq e^{-x})
    \end{align*}
\end{proof}

Equipped with the previous results, we are ready to
prove \Cref{thm:bi-criteria-distorted}.

\begin{proof}[Proof of \Cref{thm:bi-criteria-distorted}.]
    By definition of $\Phi_0(\cdot), \Phi_k(\cdot)$ we have that
    \[
        \Phi_0(S_0) = \left(1-\frac{1}{k}\right)\cdot f(\emptyset) -  c(\emptyset) = 0 \,,
    \]
    and 
    \[
        \Phi_k(S_k) = \left(1-\frac{1}{k}\right)^0 \cdot f(S_K) -  c(S_k) = f(S_k) -  c(S_k) \,.
    \]
    Notice that $\Phi_k(S_k) - \Phi_0(S_0) = \sum_{t = 0}^{k-1} \Phi_{t+1}(S_{t+1}) - \Phi_t(S_{t}).$ 
    By \Cref{lem:connection-psi-phi} it follows immediately
    that $\Phi_{t+1}(S_{t+1}) - \Phi_t(S_{t}) \geq 0.$
    Moreover, by \Cref{lem:connection-psi-phi} and \Cref{lem:psi-lower-bound} we have that
\begin{align*}
    \Phi_{t+1}(S_{t+1}) - \Phi_t(S_t) &= \Psi_t(S_t, i_t) + \frac{1}{k}\left(1- \frac{1}{k}\right)^{k-(t+1)}f(S_t) 
    & (\text{\Cref{lem:connection-psi-phi}})\\
    &\geq \frac{1}{k}\left(1-\frac{1}{k} \right)^{k-(t+1)}\left(f(\OPT) - f(S_t)\right) \\
    &\quad - \frac{1}{k} c(\OPT) + \frac{1}{k}\left(1-\frac{1}{k}\right)^{k-(t+1)}f(S_t) 
    &\ (\text{\Cref{lem:psi-lower-bound}}) \\
    &\geq \frac{1}{k}\left(1-\frac{1}{k} \right)^{k-(t+1)}f(\OPT) - \frac{1}{k} c(\OPT) 
    &\ (\text{Rearranging terms}) \\
\end{align*}
    First, assume that $\beta$ is a multiple of $1/k.$ We lower
    bound the first $k-\beta \cdot k$ terms of $\sum_{t = 0}^{k-1} \Phi_{t+1}(S_{t+1}) - \Phi_t(S_{t})$ by 0 and the last
    $\beta \cdot k$ terms by the previous inequality.
    Thus, we get 
    \begin{align*}
        \sum_{t = 0}^{k-1} \Phi_{t+1}(S_{t+1}) - \Phi_t(S_{t}) &\geq 
        \sum_{t = k-\beta \cdot k}^{k-1} \Phi_{t+1}(S_{t+1}) - \Phi_t(S_{t}) \\
        &\geq \sum_{t = k-\beta \cdot k}^{k-1}\left\{\frac{1}{k}\left(1-\frac{1}{k} \right)^{k-(t+1)}f(\OPT) - \frac{1}{k} c(\OPT) \right\} \\
        &= \left(\sum_{t = k-\beta \cdot k}^{k-1}\frac{1}{k}\left(1-\frac{1}{k} \right)^{k-(t+1)}\right)f(\OPT) - \beta \cdot c(\OPT) \\
        &= \frac{1}{k}\cdot \frac{1-\left(1-\frac{1}{k}\right)^{\beta \cdot k}}{1-\left(1-\frac{1}{k}\right)} \cdot f(\OPT) - \beta\cdot c(\OPT) & (\text{Sum of geometric series}) \\
        &\geq (1-e^{-\beta}) \cdot f(\OPT) - \beta \cdot c(\OPT). & (\text{\Cref{lem:geom-series-exp-bound}})
    \end{align*}

    In case $\beta$ is not a multiple of $1/k,$ 
    the same analysis goes through with $\hat \beta$ being the smallest
    multiple of $1/k$ that is greater than $\beta$
    and the guarantee we get is $(1-e^{-\hat \beta}, \hat \beta),$ which is at least
    $(1-e^{-\beta},  \beta+1/k).$
\end{proof}

\subsection{Stochastic Distorted Greedy Algorithm Analysis}
In this section we shift our attention to the stochastic version of the 
distorted greedy algorithm \citep{harshaw2019submodular}, which
requires fewer oracle calls to the function $f(\cdot)$ than its deterministic
counterpart.
\begin{algorithm}[h]
  \caption{Stochastic distorted greedy algorithm \citep{harshaw2019submodular}}
  \label{alg:stochastic-distorted-greedy}
  \KwData{A set of $n$ items $\N$, a monotone submodular function $f:2^\N \rightarrow \R_{\geq 0}$, a cost function $c:\N \rightarrow \R_{\geq 0},$ a capacity
  constraint $k \leq n$, an error parameter $\eps > 0$}
  \KwResult{A subset of the items $R \subseteq \N$ with $\abs{S} \leq k$}
  $S_0 = \emptyset$\;
  $s = \lceil\frac{n}{k} \log\frac{1}{\eps}\rceil$\;
  \For{$j$ from $1$ to $k$}{
  $B_j = $ sample $s$ elements uniformly and independently from $\Omega$\;
    $G(i, S_{j-1}, \bm c, j, r) = \left(1-\frac{1}{n}\right)^{k - j}\cdot f(i \mid S_{j-1}) - c _i, \forall i \in B_j$\;
    $i_t = \argmax_{i \not\in S_{j-1}, i \in B_j} G(i, S_{j-1}, \bm c, j, r)$\;
    \If{$G(i^*, S_{j-1}, \bm c, j, r) > 0$} {
      $S_{j} = S_{j-1} \cup \{i^*\}$\;
    }
    \Else{
      $S_{j} = S_{j-1}$\;
    }
  }
  \Return{$S_n$}\;
\end{algorithm}

We first state some results from \citet{harshaw2019submodular} that will be useful
in our analysis.

\begin{lemma}[\cite{harshaw2019submodular}]\label{lem:distorted-greedy-psi-lower-bound}
In each step $t$ of \Cref{alg:stochastic-distorted-greedy} it holds that
    \[
        \E[\Psi_t(S_t,i_t)] \geq (1-\eps)\left(\frac{1}{k}\left(1-\frac{1}{k}\right)^{k-(t+1)}\left(f(\OPT) - \E[f(S_t)]\right) -\frac{1}{k} c(\OPT)\right)
    \]
\end{lemma}

Notice also that we have the trivial lower bound $\E[\Psi_t(S_t,i_t)] \geq 0.$
We are now ready to state and prove our main result in this subsection.

\begin{theorem}\label{thm:bi-criteria-stochastic-distorted}
    Let $\N$ be a universe of $n$ elements, $f: 2^\N \rightarrow \R_{\geq 0}$
    be a monotone submodular function and $c: \N \rightarrow \R_{\geq 0}$
    be a cost function. Let $\OPT$ be the optimal solution of the objective
    $\max_{S \subseteq \N, \abs{S} \leq k} \{ f(S) - \sum_{i \in S} c_i\}.$
    Then, the output $R$ of \Cref{alg:stochastic-distorted-greedy} satisfies $\E\left[f(R) - \sum_{j \in R}c_j\right] \geq (1-\eps)(1-e^{-\beta})f(\OPT) - (\beta + 1/k) \sum_{j \in \OPT} c_j,$ simultaneously
    for all $\beta \in [0,1].$
\end{theorem}

\begin{proof}
        By definition of $\Phi_0(\cdot), \Phi_k(\cdot)$ we have that
    \[
        \E[\Phi_0(S_0)] = \left(1-\frac{1}{k}\right)\cdot f(\emptyset) -  c(\emptyset) = 0 \,,
    \]
    and 
    \[
       \E[ \Phi_k(S_k)] = \E\left[\left(1-\frac{1}{k}\right)^0 \cdot f(S_K) -  c(S_k)\right] = \E\left[f(S_k) -  c(S_k)\right] \,.
    \]
    Notice that, by linearity of expectation, $\E[\Phi_k(S_k) - \Phi_0(S_0)] = \sum_{t = 0}^{k-1} \E[\Phi_{t+1}(S_{t+1}) - \Phi_t(S_{t})].$ 
    By definition of $\Psi_{t+1}$ it follows immediately
    that $\E[\Phi_{t+1}(S_{t+1}) - \Phi_t(S_{t})] \geq 0.$\footnote{In fact, it holds for the realization of the random variables and not just in expectation.}
    Moreover, by \Cref{lem:connection-psi-phi} and \Cref{lem:distorted-greedy-psi-lower-bound} we have that
    \begin{align*}
    \E[\Phi_{t+1}(S_{t+1}) - \Phi_t(S_t)] 
    &= \E\left[\Psi_t(S_t, i_t) + \frac{1}{k}\left(1- \frac{1}{k}\right)^{k-(t+1)}f(S_t)\right] 
     (\text{\Cref{lem:connection-psi-phi} and linearity of expectation}) \\
    &\geq (1-\eps)\left(\frac{1}{k}\left(1-\frac{1}{k} \right)^{k-(t+1)}\left(f(\OPT) - \E[f(S_t)]\right)
     - \frac{1}{k} c(\OPT)\right)  +  \\
    &\quad \frac{1}{k}\left(1-\frac{1}{k}\right)^{k-(t+1)}\E[f(S_t)] \quad\quad\quad\quad\quad\quad\quad\quad\quad\quad\quad\quad\quad\quad\quad\quad\quad\quad\quad\quad
      (\text{\Cref{lem:psi-lower-bound}}) \\
    &\geq (1-\eps)\left( \frac{1}{k}\left(1-\frac{1}{k} \right)^{k-(t+1)}f(\OPT) - \frac{1}{k} c(\OPT)\right) +\\
    &\frac{\eps}{k}\left(1-\frac{1}{k}\right)^{k-(t+1)}\E[f(S_t)] 
\quad\quad\quad\quad\quad\quad\quad\quad\quad\quad\quad\quad\quad\quad\quad\quad\quad\quad (\text{Rearranging terms}) \\
    &\geq (1-\eps)\left( \frac{1}{k}\left(1-\frac{1}{k} \right)^{k-(t+1)}f(\OPT) - \frac{1}{k} c(\OPT)\right) 
      \quad\quad\quad\quad\quad\quad (\text{Non-negativity of } f)
\end{align*}

    First, assume that $\beta$ is a multiple of $1/k.$ We lower
    bound the first $k-\beta \cdot k$ terms of $\sum_{t = 0}^{k-1} \E[\Phi_{t+1}(S_{t+1}) - \Phi_t(S_{t})]$ by 0 and the last
    $\beta \cdot k$ terms by the previous inequality.
    Thus, we get 
    \begin{align*}
        \sum_{t = 0}^{k-1} \E[\Phi_{t+1}(S_{t+1}) - \Phi_t(S_{t})] &\geq 
        \sum_{t = k-\beta \cdot k}^{k-1} \E[\Phi_{t+1}(S_{t+1}) - \Phi_t(S_{t})] \\
        &\geq \sum_{t = k-\beta \cdot k}^{k-1}(1-\eps)\left\{\frac{1}{k}\left(1-\frac{1}{k} \right)^{k-(t+1)}f(\OPT) - \frac{1}{k} c(\OPT) \right\} \\
        &\geq (1-\eps) \left(\sum_{t = k-\beta \cdot k}^{k-1}\frac{1}{k}\left(1-\frac{1}{k} \right)^{k-(t+1)}\right)f(\OPT) - \beta \cdot c(\OPT) \quad  (\text{Non-negativity of } c)\\
        &= \frac{1-\eps}{k}\cdot \frac{1-\left(1-\frac{1}{k}\right)^{\beta \cdot k}}{1-\left(1-\frac{1}{k}\right)} \cdot f(\OPT) - \beta\cdot c(\OPT)  \quad\quad\quad\quad\quad (\text{Sum of geometric series}) \\
        &\geq (1-\eps)(1-e^{-\beta}) \cdot f(\OPT) - \beta \cdot c(\OPT).  \quad\quad\quad\quad\quad\quad\quad\quad\quad\quad (\text{\Cref{lem:geom-series-exp-bound}})
    \end{align*}

    In case $\beta$ is not a multiple of $1/k,$ 
    the same analysis goes through with $\hat \beta$ being the smallest
    multiple of $1/k$ that is greater than $\beta$
    and the guarantee we get is $((1-\eps)(1-e^{-\hat \beta}), \hat \beta),$ which is at least
    $((1-\eps)(1-e^{- \beta}),  \beta+1/k).$
\end{proof}

\subsection{Noisy Setting}\label{apx:noisy-submodular-maximization}
Following the model of \citet{horel2016maximization}, in this section we discuss adaptations of our 
results to the noisy setting where
we have access to an oracle $F:2^{\N} \rightarrow \R_{\geq 0}$
so that
\[
    (1-\eps) f(S) \leq F(S) \leq (1+\eps) f(S), \forall S \subseteq 2^\N \,,
\]
for some $\eps > 0.$ 
Our \Cref{alg:noisy-distorted-greedy} is an adaptation of \citet{harshaw2019submodular} and \citet{gong2023algorithms} with the main
difference being that when we evaluate
the score of an element we look at its
minimum marginal contribution
across all the sets we have constructed
in the history of the execution. This is
because the function $F$ is not submodular.

Similarly as in the noiseless
setting, we let
\begin{align*}
    \tilde \Phi_t(T) &= \left( 1- \frac{1}{k}\right)^{k-t}F(T) -   x\cdot \sum_{j \in T} c_j, &\forall T \subseteq \N,\\
    \tilde \Psi_t(\mathcal{T}, i) &= \max\left\{0,  \left( 1- \frac{1}{k}\right)^{k-(t+1)} \min_{S \in \mathcal{T}}F(i|S) -  x\cdot c_i \right\}, & \forall  \mathcal{T} \subseteq 2^{\N}, i \in \N.
\end{align*}

\begin{algorithm}[h]
  \caption{Noisy distorted greedy algorithm}
  \label{alg:noisy-distorted-greedy}
  \KwData{A set of $n$ items $\N$, a monotone submodular function $f:2^\N \rightarrow \R_{\geq 0}$, a cost function $c:\N \rightarrow \R_{\geq 0},$ a capacity
  constraint $k \leq n$, a cost parameter $x$}
  \KwResult{A subset of the items $R \subseteq \N$ with $\abs{S} \leq k$}
  $S_0 = \emptyset$\;
  $\S_0 = S_0$\;
  \For{$j$ from $1$ to $k$}{
    $G(i, S_{j-1}, \bm c, j, r) = \left(1-\frac{1}{n}\right)^{k - j}\cdot \min_{1\leq t \leq j}F(i \mid S_{t-1}) - x\cdot c _i, \forall i \in S$\;
    $i^* = \argmax_{i \not\in S_{j-1}} G(i, S_{j-1}, \bm c, j, r)$\;
    \If{$G(i_t, S_{j-1}, \bm c, j, r) > 0$} {
      $S_{j} = S_{j-1} \cup \{i^*\}$\;
      $\S_j = \S_{j-1} \cup S_j$\;
    }
    \Else{
      $S_{j} = S_{j-1}$\;
      $\S_j = \S_{j-1}$\;
    }
  }
  \Return{$S_n$}\;
\end{algorithm}

The next result describes the connection 
between $\tilde \Phi, \tilde \Psi.$ Our proof 
is an adaptation of \citet{harshaw2019submodular, gong2023algorithms}.
\begin{lemma}\label{lem:noisy-connection-psi-phi}
    In each iteration $t \in \{0,\ldots,k-1\}$ of \Cref{alg:noisy-distorted-greedy} 
    it holds that
    \[
        \tilde \Phi_{t+1}(S_{t+1}) - \tilde \Phi_{t}(S_{t}) \geq \tilde \Psi_t(\S_t,i_t) + \frac{1}{k}\left(1- \frac{1}{k}\right)^{k-(t+1)}F(S_t) \,,
    \]
    where $\S_t = \{S_0,\ldots,S_t\}.$
\end{lemma}

\begin{proof}
    By definition we have that
    \begin{align*}
        \tilde \Phi_{t+1}(S_{t+1}) - \tilde \Phi_{t}(S_{t}) &= 
        \left(1-\frac{1}{k}\right)^{k-(t+1)}F(S_{t+1}) - c(S_{t+1}) - \left(1-\frac{1}{k}\right)^{k-t}F(S_{t}) - c(S_{t}) \\
         &= 
        \left(1-\frac{1}{k}\right)^{k-(t+1)}F(S_{t+1}) - c(S_{t+1}) - \left(1-\frac{1}{k}\right)^{k-(t+1)}\left(1-\frac{1}{k}\right)F(S_{t}) - c(S_{t}) \\
         &= \left(1-\frac{1}{k}\right)^{k-(t+1)}(F(S_{t+1})- F(S_{t})) - (c(S_{t+1}) - c(S_t)) + \frac{1}{k}\left(1-\frac{1}{k}\right)^{k-(t+1)} F(S_t) \,.
    \end{align*}
    Now we consider two cases. If $S_{t+1} = S_t$
    then $\tilde \Psi(\mathcal{S}_t, i_t) = 0$
    and the inequality holds. Otherwise, we have
    \begin{align*}
        \left(1-\frac{1}{k}\right)^{k-(t+1)}(F(S_{t+1})- F(S_{t})) - (c(S_{t+1}) - c(S_t)) + \frac{1}{k}\left(1-\frac{1}{k}\right)^{k-(t+1)} F(S_t) &= \\
        \left(1-\frac{1}{k}\right)^{k-(t+1)} F(i_t |S_t) - c_{i_t} + \frac{1}{k}\left(1-\frac{1}{k}\right)^{k-(t+1)} F(S_t) &\geq \\
        \left(1-\frac{1}{k}\right)^{k-(t+1)} \min_{1\leq j\leq t} F(i_t |S_j) - c_{i_t} + \frac{1}{k}\left(1-\frac{1}{k}\right)^{k-(t+1)} F(S_t) &= \\
         \tilde \Psi_t(\S_t,i_t) + \frac{1}{k}\left(1- \frac{1}{k}\right)^{k-(t+1)}F(S_t) \,.
    \end{align*}
\end{proof}

The next result relates the marginal gain in each iteration
of \Cref{alg:noisy-distorted-greedy} to the function $\tilde \Psi_t(\S_t,i_t).$ It builds upon the approach of
\citet{harshaw2019submodular,gong2023algorithms}.
\begin{lemma}\label{lem:noisy-psi-lower-bound}
    In each iteration of \Cref{alg:noisy-distorted-greedy} it holds that
    \[
       \tilde \Psi_t(\S_t,i_t) \geq \frac{1-\eps}{k}\left( 1- \frac{1}{k}\right)^{k-(t+1)}\left(f(\OPT) - f(S_t)\right) - 2\eps  \left( 1- \frac{1}{k}\right)^{k-(t+1)} f(S_k) - \frac{x}{k}c(\OPT)  \,.
    \]
\end{lemma}

\begin{proof}
    First, notice that 
    \begin{equation}\label{eq:1}
        F(S|i) = F(S \cup \{i\}) - F(S) \geq (1-\eps) f(S\cup \{i\}) - (1+\eps) f(S) = (1-\eps)f(S|i) - 2\eps f(S) \,.
    \end{equation}
     Let $\OPT$ be the optimal solution of $\max_S \{f(S) - c(S)\}$. We have that
    \begin{align*}
        k \tilde \Psi_t(\S_t,i_t) &= k \cdot  \max\left\{0,  \left( 1- \frac{1}{k}\right)^{k-(t+1)} \min_{S \in {\S}_t}F(i_t|S) -  x\cdot c_i \right\} &  (\text{Definition})\\
        &\geq k \cdot \left(\left( 1- \frac{1}{k}\right)^{k-(t+1)} \min_{S \in {\S}_t}F(i_t|S) -  x\cdot c_i \right) & (\text{Restricting max}) \\
         &= k \cdot \max_{i \in \N}\left\{\left( 1- \frac{1}{k}\right)^{k-(t+1)} \min_{S \in {\S}_t}F(i|S) -  x\cdot c_i \right\} & (\text{Definition})\\
         &\geq \abs{\OPT} \cdot \max_{i \in \N}\left\{\left( 1- \frac{1}{k}\right)^{k-(t+1)} \min_{S \in {\S}_t}F(i|S) -  x\cdot c_i \right\} & (k \geq \abs{\OPT}) \\
         &\geq \abs{\OPT} \cdot \max_{i \in \OPT}\left\{\left( 1- \frac{1}{k}\right)^{k-(t+1)} \min_{S \in {\S}_t}F(i|S) -  x\cdot c_i \right\} & (\text{Restricting max}) \\
          &\geq \sum_{i \in \OPT}\left\{\left( 1- \frac{1}{k}\right)^{k-(t+1)} \min_{S \in {\S}_t}F(i|S) -  x\cdot c_i \right\} & (\text{Averaging argument})\\
          &= \sum_{i \in \OPT}\left\{\left( 1- \frac{1}{k}\right)^{k-(t+1)} F(i|S^i) -  x\cdot c_i \right\} & (S^i = \argmin_{S \in \S_t} F(i|S))\\
          &\geq \sum_{i \in \OPT}\left\{\left( 1- \frac{1}{k}\right)^{k-(t+1)} \left((1-\eps)f(i|S^i) - 2\eps f(S^i)\right)- x\cdot c_i \right\} & (\text{\Cref{eq:1}})\\
          &\geq \sum_{i \in \OPT}\left\{\left( 1- \frac{1}{k}\right)^{k-(t+1)} \left((1-\eps)f(i|S_t) - 2\eps f(S^i)\right)- x\cdot c_i \right\} & (\text{Submodularity of } f)\\
          &\geq \sum_{i \in \OPT}\left\{\left( 1- \frac{1}{k}\right)^{k-(t+1)} (1-\eps)f(i|S_t) - 2\eps \left( 1- \frac{1}{k}\right)^{k-(t+1)}f(S_k)- x\cdot c_i \right\} \,. & (\text{Monotonicity of } f) 
    \end{align*}

    We will bound each term of the summation separately.
    Notice that $\sum_{i \in \OPT} x\cdot c_i = x\cdot c(\OPT).$ Similarly, $\sum_{i \in \OPT}2\eps \left( 1- \frac{1}{k}\right)^{k-(t+1)}f(S_k) \leq 2\eps\abs{\OPT}f(S_k) \leq  2\eps k f(S_k),$ and by the submodularity and monotonicity of $f$ we have
    \[
        \sum_{i \in \OPT}(1-\eps)f(i|S_t) \geq (1-\eps)\left(f(S_t \cup \OPT) - f(S_t)\right) \geq (1-\eps)\left(f(\OPT) - f(S_t)\right)  \,.
    \]
    Putting everything together, we get
    \begin{align*}
         k \tilde \Psi_t(\S_t,i_t) \geq (1-\eps)\left( 1- \frac{1}{k}\right)^{k-(t+1)}\left(f(\OPT) - f(S_t)\right) - 2\eps k \left( 1- \frac{1}{k}\right)^{k-(t+1)} f(S_k) - xc(\OPT) \,,
    \end{align*}
    and dividing by $k$ we get
    \begin{align*}
         \tilde \Psi_t(\S_t,i_t) \geq \frac{1-\eps}{k}\left( 1- \frac{1}{k}\right)^{k-(t+1)}\left(f(\OPT) - f(S_t)\right) - 2\eps  \left( 1- \frac{1}{k}\right)^{k-(t+1)} f(S_k) - \frac{x}{k}c(\OPT) \,.
    \end{align*}
\end{proof}

We are now ready to prove the bi-criteria
guarantees of \Cref{alg:noisy-distorted-greedy}.
Our analysis follows \citet{harshaw2019submodular,gong2023algorithms}.

\begin{theorem}\label{thm:noisy-distorted-greedy-main-result}
    For $x = 1 + 2\eps k+ \eps$ \Cref{alg:noisy-distorted-greedy}
    returns a set $S_k$ with
    \[
        f(S_k) - c(S_k) \geq  \frac{1-\eps}{ 1 + 2\eps k+ \eps} \cdot\left( 1- 1/e\right) f(\OPT) - c(S_k) \,.
    \]
\end{theorem}

\begin{proof}
    Combining \Cref{lem:noisy-connection-psi-phi} 
    and \Cref{lem:noisy-psi-lower-bound} we 
    immediately get that
    \begin{align*}
        &\tilde \Phi_{t+1}(S_{t+1}) - \tilde \Phi_{t}(S_{t}) \geq \\
        &\tilde \Psi_t(\S_t,i_t) + \frac{1}{k}\left(1- \frac{1}{k}\right)^{k-(t+1)}F(S_t) \geq  \\
        &\frac{1-\eps}{k}\left( 1- \frac{1}{k}\right)^{k-(t+1)}\left(f(\OPT) - f(S_t)\right) - 2\eps  \left( 1- \frac{1}{k}\right)^{k-(t+1)} f(S_k) - \frac{x}{k}c(\OPT)  + \frac{1}{k}\left(1- \frac{1}{k}\right)^{k-(t+1)}F(S_t) \geq \\
         &\frac{1-\eps}{k}\left( 1- \frac{1}{k}\right)^{k-(t+1)}\left(f(\OPT) - f(S_t)\right) - 2\eps  \left( 1- \frac{1}{k}\right)^{k-(t+1)} f(S_k) - \frac{x}{k}c(\OPT)  + \frac{1-\eps}{k}\left(1- \frac{1}{k}\right)^{k-(t+1)}f(S_t) =\\
         &\frac{1-\eps}{k}\left( 1- \frac{1}{k}\right)^{k-(t+1)} f(\OPT)- 2\eps  \left( 1- \frac{1}{k}\right)^{k-(t+1)} f(S_k) - \frac{x}{k} c(OPT) \,.
    \end{align*}
    Moreover, notice that a straightforward bound
    is
    \[
        \tilde \Phi_{t+1}(S_{t+1}) - \tilde \Phi_{t}(S_{t}) \geq 0 \,.
    \]
    By definition of $\tilde \Phi$ we have that
    \begin{align*}
       \tilde \Phi_0(S_0) = \left(1-\frac{1}{k}\right)^k F(\emptyset) - xc(\emptyset) = 0 \\
       \tilde  \Phi_k(S_k) = \left(1-\frac{1}{k}\right)^k F(S_k) - xc(S_k) \,.
    \end{align*}
    Combining the previous inequalities we 
    get
    \begin{align*}
        F(S_k) - x c(S_k) \\
        &= \tilde \Phi_k(S_k) - \tilde \Phi_0(S_0) \\
        &= \sum_{i=1}^k (\Phi_i(S_i) - \tilde \Phi_{i-1}(S_{i-1})) \\
        &\geq \sum_{i=1}^k \left(\frac{1-\eps}{k}\left( 1- \frac{1}{k}\right)^{k-(t+1)} f(\OPT)- 2\eps  \left( 1- \frac{1}{k}\right)^{k-(t+1)} f(S_k) - \frac{x}{k} c(OPT)\right) \\
        &= \frac{1-\eps}{k} \sum_{i=1}^k \left( 1- \frac{1}{k}\right)^{k-(t+1)} f(\OPT) - 2\eps\sum_{i=1}^k \left( 1- \frac{1}{k}\right)^{k-(t+1)} f(S_k) - x\cdot c(\OPT) \\
        &= \frac{1-\eps}{k} \cdot k \cdot\left( 1- (1-1/k)^k\right) f(\OPT) - 2\eps  k \cdot\left( 1- (1-1/k)^k\right)f(S_k) - x\cdot c(\OPT) \,,
    \end{align*}
    which implies that
    \begin{align*}
        (1+\eps) f(S_k) - x c(S_k) &\geq (1-\eps) \cdot\left( 1- (1-1/k)^k\right) f(\OPT) -2\eps  k \cdot\left( 1- (1-1/k)^k\right)f(S_k) - xc(\OPT) \\
        &\geq (1-\eps) \cdot\left( 1- (1-1/k)^k\right) f(\OPT) -2\eps  k \cdot f(S_k) - xc(\OPT) \,,
    \end{align*}
    and rearranging we get
    \begin{align*}
        (1+2\eps k+ \eps) f(S_k) - x c(S_k) &\geq  (1-\eps) \cdot\left( 1- (1-1/k)^k\right) f(\OPT) - x c(\OPT) \\
         &\geq  (1-\eps) \cdot\left( 1- 1/e\right) f(\OPT) - x c(\OPT) \,.
    \end{align*}
    Finally, we can simplify $x = 1 + 2\eps k+ \eps$ and get
    \[
        f(S_k) - c(S_k) \geq \frac{1-\eps}{ 1 + 2\eps k+ \eps} \cdot\left( 1- 1/e\right) f(\OPT) - c(S_k) \,.
    \]
\end{proof}

\begin{remark}[IC, IR, NAS in the Noisy Setting]\label{rem:IC-IR-NAS-noisy}
    We remark that our modification in \Cref{alg:noisy-distorted-greedy} enforces
    the diminishing returns property in the scores of
    the elements that are added to the constructed solutions. Hence, an identical argument to the one we used in the proof of \Cref{thm:meta-algo-mech} shows that the NAS property is satisfied.
    The IC, IR properties continue to hold since
    they are not affected by the submodularity of the function.
\end{remark}

\section{Omitted Details from \Cref{sec:sealed-bid-framework}}\label{apx:sealed-bid-framweork}
Let $\OPT(\bm b) \in \argmax_{S \in 2^{\N}} f(S) - \sum_{i \in S} c_i$. Given a bid profile $\bm b$, the VCG mechanism purchases items from the sellers in $\OPT(\bm b)$ and the payment $p_i$ to each $i \in \OPT(\bm b)$ is
given by
\begin{align}
    p_i = \left(f\big(\OPT(\bm b)\big) - \sum_{j \in \OPT(\bm b) \setminus \{i\}} c_j\right) - \left(f\Big(\OPT\big((\infty, \bm b_{-i})\big)\Big) - \sum_{j \in \OPT\big((\infty, \bm b_{-i})\big)} c_j\right)\,. \nonumber
\end{align}
\begin{proof}[Proof of \Cref{prop:vcg-nas}]
    We would like to show $f\big(\OPT(\bm b)\big) \geq \sum_{i \in \OPT(\bm b)} p_i$ when $f$ is a submodular function.
    \begin{align}
        \sum_{i \in \OPT(\bm b)} p_i &~= \sum_{i \in \OPT(\bm b)} \left(f\big(\OPT(\bm b)\big) - \sum_{j \in \OPT(\bm b) \setminus \{i\}} c_j\right) - \left(f\Big(\OPT\big((\infty, \bm b_{-i})\big)\Big) - \sum_{j \in \OPT\big((\infty, \bm b_{-i})\big)} c_j\right) \nonumber \\
        &~\leq \sum_{i \in \OPT(\bm b)} \left(f\big(\OPT(\bm b)\big) - \sum_{j \in \OPT(\bm b) \setminus \{i\}} c_j\right) - \left(f\big(\OPT(\bm b) \setminus \{i\}\big) - \sum_{j \in \OPT(\bm b) \setminus \{i\}} c_j\right) \nonumber \\
        &~= \sum_{i \in \OPT(\bm b)} f\big(\OPT(\bm b)\big) - f\big(\OPT(\bm b) \setminus \{i\}\big) \nonumber \\
        &~\leq f\big(\OPT(\bm b)\big), \nonumber
    \end{align}
    where the first inequality uses the fact that $\OPT\big((\infty, \bm b_{-i})\big) \in \argmax_{S: i \not\in S} f(S) - \sum_{j \in S} c_j$ and the second inequality follows the property of submodularity.
\end{proof}

\begin{proof}[Proof of \Cref{thm:meta-algo-mech}]
    Fix $\bm b_{-i}$ and $r$, and let $\left\{S_0^{b_i}, S_1^{b_i}, \cdots, S_n^{b_i}\right\}$ be the intermediate tentative solutions of running $\A\big(\N, (b_i, \bm b_{-i}), r\big)$. Moreover, let $p_{i, k} = \max_{1 \leq j \leq k} z_j^*$ where
    \begin{align*}
        z_j^* = \sup\left\{z \mid i = \argmax_{\ell \not\in S_{j-1}^\infty} G\left(\ell, S_{j-1}^{\infty}, (z, \bm b_{-i}), j, r\right) ~\&~ G\left(i, S_{j-1}^{\infty}, (z, \bm b_{-i}), j, r\right) > 0 \right\}.
    \end{align*}
    We will show that for any $i$, $\bm b_{-i}$, $r$, and $k$: (1) if $b_i < p_{i, k}$, $i \in S_k^{b_i}$, and if $b_i > p_{i, k}$, $i \not\in S_k^{b_i}$; (2) if $i \not\in S_k^{b_i}$, $S_k^{b_i} = S_k^{\infty}$. We prove it by induction on $k$. When $k = 1$, from Algorithm~\ref{alg:meta-algo}, the definition of $p_{i,1}$, and Assumption~\ref{assump:meta-algo}(1), we have 
    \begin{itemize}
        \item when $b_i < p_{i,1}$, then $S_1^{b_i} = \{i\}$ since $i = \arg\max_{\ell} G\big(\ell, \emptyset, (b_i, \bm b_{-i}), 1, r\big)$;
        \item when $b_i > p_{i,1}$, with $\ell^*_1 = \arg\max_{\ell \neq i} G\big(\ell, \emptyset, (b_i, \bm b_{-i}), 1, r\big)$, if $G\big(\ell, \emptyset, (b_i, \bm b_{-i}), 1, r\big) > 0$, $S_1^{b_i} = \{\ell^*_1\}$; otherwise, $S_1^{b_i} = \emptyset$.
    \end{itemize}
    Moreover, notice that whenever $i \neq \ell^*_1$, $i$'s bid becomes irrelevant due to Assumption~\ref{assump:meta-algo}(3) so that 
    \begin{align*}
        \ell^*_1 = \arg\max_{\ell} G\big(\ell, \emptyset, (\infty, \bm b_{-i}), 1, r\big)\,,
    \end{align*}
    and therefore, if $i \not\in S_1^{b_i}$, $S_1^{b_i} = S_1^\infty$.
    For the inductive step, we assume the previous arguments hold for all rounds up to $k$. Then, for round $k+1$, we have
    \begin{itemize}
        \item when $b_i < p_{i,k+1}$, then either we have $i \in S_k^{b_i} \subset S_{k+1}^{b_i}$ or we have $S_k^{b_i} = S_k^\infty$, which also implies $i \in S_{k+1}^{b_i}$ since $i = \arg\max_{\ell \not\in S_k^\infty} G\big(\ell, S_k^\infty, (b_i, \bm b_{-i}), k+1, r\big)$ and $G\big(i, S_k^\infty, (b_i, \bm b_{-i}), k+1, r\big) > 0$;
        \item when $b_i > p_{i,k+1}$, with $\ell^*_{k+1} = \arg\max_{\ell \neq i} G\big(\ell, S_k^\infty, (b_i, \bm b_{-i}), k+1, r\big)$, if $G\big(\ell^*_{k+1}, S_k^\infty, (b_i, \bm b_{-i}), k+1, r\big) > 0$, $S_{k+1}^{b_i} = S_k^\infty \cup \{\ell^*_{k+1}\}$; otherwise, $S_{k+1}^{b_i} = S_k^\infty$.
    \end{itemize}
    Again, notice that whenever $i \neq \ell^*_{k+1}$, $i$'s bid becomes irrelevant due to Assumption~\ref{assump:meta-algo}(3) so that 
    \begin{align*}
        \ell^*_{k+1} = \arg\max_{\ell \not\in S_k^\infty} G\big(i, S_k^\infty, (\infty, \bm b_{-i}), k+1, r\big);
    \end{align*}
    and therefore, if $i \not\in S_{k+1}^{b_i}$, $S_{k+1}^{b_i} = S_{k+1}^\infty$, which concludes the inductive step.
    Observe that Algorithm~\ref{alg:meta-algo-payment} exactly computes the critical bid $p_i = p_{i,n}$ for seller $i$ such that if $b_i < p_{i,n}$, $i \in S^*$ and if $b_i > p_{i,n}$, $i \not\in S^*$. From Myerson's lemma~\citep{myerson1981optimal}, the mechanism is IC for seller $i$. Moreover, the mechanism is IR because $i \in S^*$ only if $c_i = b_i \leq p_{i,n}$ when seller $i$ reports truthfully.

    Finally, to prove the mechanism satisfies NAS, assume seller $i$ is added to the solution in round $k$ such that $S_k^{b_i} \setminus S_{k-1}^{b_i} = \{i\}$. From Assumption~\ref{assump:meta-algo}(2), we have $b_i \leq f(S_k^{b_i}) - f(S_{k-1}^{b_i})$ in order to have a positive score for $i$ in round $k$. As we have established that if $b_i < p_{i,k-1}$, then $i \in S_{k-1}^{b_i}$, we have $p_{i,k-1} \leq b_i \leq f(S_k^{b_i}) - f(S_{k-1}^{b_i})$. For $j \geq k$, we have $z_j^* \leq f(i \mid S_{j-1}^\infty)$ due to the fact that $G(i, S_{j-1}, \bm b, j, r) < 0$ whenever $b_i > f(i \mid S_{j-1}^\infty)$ from Assumption~\ref{assump:meta-algo}(2). From the previously proved fact that if $i \not\in S_{k-1}^{b_i}$, $S_{k-1}^{b_i} = S_{k-1}^\infty$, we have $S_{k-1}^{b_i} \subseteq S_{j-1}^\infty$ for $j \geq k$ by submodularity of $f$. As a result, for $j \geq k$, we have that
    \begin{align*}
    z_j^* &\leq f(i \mid S_{j-1}^\infty) \leq f(i \mid S_{k-1}^{b_i}) = f(S_k^{b_i}) - f(S_{k-1}^{b_i})\,,
    \end{align*}
    which implies that
    $p_i = \max\left\{p_{i, k-1}, \max_{k \leq j \leq n} z_j^*\right\} \leq f(S_k^{b_i}) - f(S_{k-1}^{b_i})$. Thus,
    \begin{align*}
        \sum_{i \in S_n^{b_i}} p_i \leq \sum_{k=1}^n f(S_k^{b_i}) - f(S_{k-1}^{b_i}) = f(S_n^{b_i}).
    \end{align*}
\end{proof}

\subsection{Omitted Details from 
Online Mechanism Design Framework}\label{apx:online-setting}

\begin{proof}[Proof of \Cref{thm:meta-algo-online-mech}]
    The IC and IR properties follow immediately from the fact that the mechanism
    is a posted-price mechanism, such that seller $k$ accepts posted price $\hat p_k$ if and only if $\hat p_k > c_k$\footnote{For simplicity, we assume that the seller does not accept the offer when the cost is exactly the same as the posted price.}. From the definition of $\hat p_k$ and the monotonicity of $G$ from Assumption \ref{assump:meta-algo-online}(1), we have $G(k, S_{k-1}, \bm c_{(1,k)}, r) > 0$ if and only if $c_k < \hat p_k$, and therefore, Algorithm \ref{alg:posted-price-mechanism} and Algorithm \ref{alg:meta-algo-online} return the same solution if sellers always best respond to the posted prices.

    Finally, we prove for NAS. For seller $k \in S^*$, from Assumption \ref{assump:meta-algo-online}(2), we have 
    \[
         \hat p_k \leq f(k \mid S_{k-1}) = f(S_k) - f(S_{k-1})\,.
    \]
    Thus, summing up over all sellers in $S^*$, we have
    \[
        \sum_{k \in S^*} p_{k} \leq \sum_{k \in S^*} f(S_k) - f(S_{k-1}) = f(S^*).
    \]
\end{proof}

\section{Omitted Details from \Cref{sec:descending-auctions}}\label{apx:descending-auctions}

\begin{proof}[Proof of \Cref{thm:descend-negative}]
    We construct an instance with $L+2$ sellers, indexed by $\{1, \cdots, L, L+1, L+2\}$, with a submodular function $f$ such that: if $\{L+1, L+2\} \cap S = \emptyset$, $f(S) = |S|$; otherwise, $f(S) = L$.
    Moreover, let $b_i = 1/L$ for $i \leq L$ and $b_{L+1} = b_{L+2} = L-2$. We next describe the strategy of the adversary for selecting seller $i \in S \setminus \D(S, \bm p)$ in line 4 of Algorithm~\ref{alg:descending-auctions}. 
    
    First, the adversary keeps selecting a seller $j \in \{L+1, L+2\}$ until either $p_{L+1} < L-1$ or $p_{L+2} < L-1$. This is achievable since if $\{L+1, L+2\} \subseteq \D(S, \bm p)$, the welfare is at most $L - 2 \times (L-1) < 0$,
    so that the approximation guarantee is violated. Without loss of generality, assume that $p_{L+1} < L-1$, and therefore, the welfare of selecting seller $L+1$ alone is $f(\{L+1\}) - p_{L+1} > 1$.
     
    Next, the adversary iterates over sellers from seller $1$ to seller $L$ such that for each seller $i$, keep selecting the seller $i$ until $i$ leaves the market, i.e., $p_i < b_i$. We argue that the above process is achievable as the welfare of selecting any subset containing seller $i$ is at most $1$.

    The final welfare is at most $2$ by purchasing items from either seller $L+1$ or $L+2$. However, the optimal welfare is obtained by purchasing from sellers in $\{1, \cdots, L\}$, which gives welfare $L - 1$.
\end{proof}

\begin{proof}[Proof of \Cref{thm:descending-auction-adversarial}]
    The $(\frac 12, 1)$-approximate demand oracle $\hat \D$ we construct maintains a tentative solution $S$ initialized at $S = \emptyset$. For each iteration, let $i$ be the candidate selected by the adversary in the previous iteration. Update $S = S \cup \{i\}$ if $f(i \mid S) > 2 p_i$ (otherwise, keep $S$ as it is), and return $S$. Let $S^*$ be the subset returned by the Algorithm~\ref{alg:descending-auctions} with demand oracle $\hat \D$. For each seller $i \in S^*$, let $\hat p_i$ and $S_i$ be the price $p_i$ and the tentative solution $S$ maintained by the demand oracle right before $i$ is added to the solution, respectively. For seller $i \not\in S^*$, let $\hat p_i$ and $S_i$ be the price $p_i$ and the tentative solution $S$ maintained by the demand oracle when $i$ is removed from the descending auction, respectively.

    Let $g(S) = f(S) - 2 \cdot \sum_{i \in S} b_i$, which is also a submodular function. For any seller $i \in \OPT \setminus S^*$, 
    \[
        g(i \mid S^*) = f(i \mid S^*) - 2 \cdot b_i \leq f(i \mid S^*) - 2 \cdot \hat p_i \leq 2 \varepsilon
    \]
    where the first inequality follows $b_i - \varepsilon \leq \hat p_i < b_i$ for $i \not\in S^*$ and the second inequality follows the fact that $f(i \mid S^*) \leq f(i \mid S_i) \leq 2 \cdot (\hat p_i + \varepsilon)$. As a result, we have
    \begin{align*}
        &{}2 \varepsilon \cdot |\OPT \setminus S^*| \geq \sum_{i \in \OPT \setminus S^*} g(i \mid S^*) \geq g(S^* \cup \OPT) - g(S^*) \\
        ={}&{} \big(f(S^* \cup \OPT) - f(S^*)\big) - 2 \cdot \sum_{i \in \OPT \setminus S^*} b_k \\
        \geq{}&{} f(\OPT) - f(S^*) - 2 \cdot \sum_{i \in \OPT} b_i,
    \end{align*}
    where the second inequality follows submodularity of $g$. Rearranging, we obtain
    \[
        f(S^*) \geq f(\OPT) - 2 \cdot \sum_{i \in \OPT} b_i - 2 \varepsilon \cdot |\OPT \setminus S^*| \geq f(\OPT) - 2 \cdot \sum_{i \in \OPT} b_i - 2 n \varepsilon.
    \]
     Observe that we have
    \begin{align*}
        f(S^*) - 2 \cdot \sum_{i \in S^*} b_i = \sum_{i \in S^*} f(i \mid S_i ) - 2 \cdot b_{i} \geq \sum_{i \in S^*} f(i \mid S_i) - 2 \cdot \hat p_i \geq 0,
    \end{align*}
    where the first inequality follows $\hat p_i \geq b_i$ for $i \in S^*$ and the second inequality is due to $i \in S^*$ and the definition of $\hat p_i$ and $S_i$. Therefore, we have $\sum_{i \in S^*} b_i \leq \frac 1 2 \cdot f(S^*)$, indicating $f(S^*) - \sum_{i \in S^*} b_i \geq \frac 1 2 \cdot f(S^*)$. Putting everything together, we have
    \[
        f(S^*) - \sum_{i \in S^*} b_i \geq \frac 1 2 \cdot f(S^*) \geq \frac 1 2 \cdot f(\OPT) - \sum_{i \in \OPT} b_i - n \varepsilon.
    \]
\end{proof}

\begin{algorithm}[h]
  \caption{Descending auction construction for a given meta algorithm $\A^o$ and a step size of $\varepsilon$}
  \label{alg:online-algo-to-descending-auctions}
  \KwData{A set of seller $\N$ and a bid profile $\bm b$ from sellers}
  \KwResult{A subset $S^*$ of sellers to purchase from and a vector $\bm p$ of payment to sellers}
  Generate a random seed $r$ if needed or set $r = 0$\;
  $S_0 = \emptyset$\;
  Set initial prices as $p_i = f(i \mid \emptyset)$\;
  \For{$k$ from $1$ to $n$}{
    Let $\hat p_k$ be the unique solution of the equation $G\Big(k, S_{k-1}, \big(\bm c_{(1,k-1)}, z\big), r\Big) = 0$ in terms of $z$\;
    Update $p_k = \hat p_k$\;
    \If{$p_k > b_i$} {
        $S_k = S_{k-1} \setminus \{i\}$\;
    } \Else {
        $S_k = S_{k-1}$\;
        $p_k = 0$\;
    }
  }
  \Return{$S$ and $\bm p$}\;
\end{algorithm}

\section{Omitted Details from \Cref{sec:experiment}}\label{apx:experiments}

\paragraph{Heuristic Implementation}
We provide a heuristic to optimize the running time of Algorithm~\ref{alg:meta-algo} if the scoring rule of a meta algorithm $\A = (G)$ has a diminishing-return structure, i.e., $G(i, S, \bm b, j, r) \geq G(i, T, \bm b, k, r)$ for all $S \subseteq T$ and $j \leq k$. In particular, inspired by the lazy implementation of the classical greedy algorithm \citep{minoux2005accelerated}, we maintain a priority queue that records each candidate's {\em last-updated score}. The priority queue is initialized consisting of elements with key $i$ and value $G(i, \emptyset, \bm c, 0, r)$ for each $i \in \N$. For each iteration $k$, we repeatedly compare candidate $i_1$ with the highest score in $Q$ and candidate $i_2$ with the second-highest score in $Q$. If $G(i_1, S, \bm c, k, r)$ for the tentative solution $S$ is positive and is larger than the $i_2$'s score maintained in $Q$, then $i_1$ is guaranteed to have the highest score due to the diminishing-return structure; otherwise, update the score for $i_1$ and repeat (see Algorithm~\ref{alg:meta-algo-fast} in Appendix). The running time of the payment calculation in Algorithm~\ref{alg:meta-algo-payment} can be optimized in a similar way (see Algorithm~\ref{alg:meta-algo-payment-fast} in Appendix for details). Note that the scoring rules of all deterministic algorithms presented in Table~\ref{tab:mechanism-guarantees-different-submodular-algorithm} have a diminishing-return structure, except for the distorted greedy algorithm.

\section{Faster Implementations of Algorithm~\ref{alg:meta-algo} and Algorithm~\ref{alg:meta-algo-payment}}

We present the algorithmic descriptions of faster implementations of Algorithm~\ref{alg:meta-algo} and Algorithm~\ref{alg:meta-algo-payment} when the scoring rule of the meta-algorithm has a diminishing-return structure.

\begin{algorithm}[h]
  \caption{A faster implementation of Algorithm~\ref{alg:meta-algo} when $G$ has a diminishing-return structure}
  \label{alg:meta-algo-fast}
  \KwData{A set of seller $\N$, a cost profile $\bm c$ from sellers, and a random seed $r$}
  \KwResult{A subset of sellers to purchase services from}
  $S_0 = \emptyset$\;
  Initialize an empty descending-order priority queue $Q$\;
  \For{$i$ from $1$ to $n$}{
    Insert an element to $Q$ with key $i$ and value $G(i, \emptyset, \bm c, 0, r)$\;
  }
  \While{$k < n$} {
    Pop the highest score candidate $i^*$ from $Q$\;
    \While{True}{
        Let the highest score from $Q$ be $s^*$ (without popping the candidate)\;
        \If{$G(i^*, S_k, \bm c, k+1, r) > \max(0, s^*)$ or $s^* < 0$} {
            Break\;
        } \Else {
            Pop the highest score candidate $j^*$ from $Q$\;
            Insert an element to $Q$ with key $i^*$ and value $G(i^*, S_k, \bm c, k+1, r)$\;
            $i^* = j^*$\;
        }
    }
    \If{$G(i^*, S_k, \bm c, k+1, r) > 0$} {
      $S_{k+1} = S_k \cup \{i^*\}$\;
      $k = k + 1$\;
    }
    \Else{
      break\;
    }
  }
  \Return{$S_k$}\;
\end{algorithm}

\begin{algorithm}[h]
  \caption{A faster implementation of Algorithm~\ref{alg:meta-algo-payment} when $G$ has a diminishing-return structure}
  \label{alg:meta-algo-payment-fast}
  \KwData{A set of sellers $\N$, a bid profile $\bm b$ from sellers, and a meta algorithm $\A$}
  \KwResult{A subset of sellers to purchase from and a vector of payment to sellers}
  Generate a random seed $r$ if needed or set $r = 0$\;
  $S^* = \A(\N, \bm b, r)$ computed using Algorithm~\ref{alg:meta-algo-fast} and record the intermediate solutions $\{S_0, S_1, \cdots, S_K\}$\;
  \For{$k$ from $1$ to $K$} {
    $i = S_k \setminus S_{k-1}$, $T_0 = S_{k-1}$, and $j = 0$\;
    Initialize an empty descending-order priority queue $Q$\;
    \For{$i \in \N \setminus S_k$}{
      Insert an element to $Q$ with key $i$ and value $G(i, S_{k-1}, \bm c, 0, r)$\;
    }

    $p_i = b_i$\;
    \While{$j < n - k$} {
      Pop the highest score candidate $\ell^*$ from $Q$\;
        \While{True}{
            Let the highest score from $Q$ be $s^*$ (without popping the candidate)\;
            \If{$G(\ell^*, T_j, \bm c, j+k, r) > \max(0, s^*)$ or $s^* < 0$} {
                Break\;
            } \Else {
                Pop the highest score candidate $\ell'$ from $Q$\;
                Insert an element to $Q$ with key $\ell^*$ and value $G(\ell^*, T_j, \bm b, j+k, r)$\;
                $\ell^* = \ell'$\;
            }
        }
        \If{$G(\ell^*, T_j, \bm c, j+k, r) > 0$} {
          $p_i = \max\left(p_i, \sup\left\{z \geq 0 \mid G\big(i, T_j, (z, \bm b_{-i}), j+k, r\big) > G(\ell^*, T_j, \bm b, j+k, r)\right\}\right)$\;
          $T_{j+1} = T_j \cup \{\ell^*\}$\;
          $j = j + 1$\;
        }
        \Else{
          break\;
        }
    }
  }
  \Return{$S^*$ and $\bm p$}\;
\end{algorithm}

\end{document}